\documentclass[journal,twoside]{IEEEtran}
\usepackage{pifont}
\usepackage{textcomp}
\usepackage{stfloats}
\usepackage{subfloat}
\usepackage{url}
\usepackage{verbatim}
\usepackage{cite}
\usepackage{amsmath}

\usepackage{multirow}
\usepackage{amsmath,amssymb,amsfonts}
\usepackage{algorithmicx,algorithm}
\usepackage[noend]{algpseudocode}
\usepackage{textcomp}
\usepackage{xcolor}
\usepackage{caption}
\usepackage{array}
\usepackage{stfloats}
\usepackage{bm}
\usepackage{harpoon}
\usepackage{extarrows}
\usepackage{booktabs}
\usepackage{graphicx}
\usepackage{hyperref}
\usepackage{bbm}
\usepackage{subfig}
\newtheorem{lemma}{Lemma}

\def\BibTeX{{\rm B\kern-.05em{\sc i\kern-.025em b}\kern-.08em
    T\kern-.1667em\lower.7ex\hbox{E}\kern-.125emX}}
\usepackage{balance}

\newcommand{\cmark}{\ding{51}} 
\newcommand{\xmark}{\ding{55}} 

\begin{document}
\title{{{Multi-Stream Spatiotemporal Channel Coding for \\
			 MIMO Systems: Transmission Scheme Design \\ and
			  Achievable Rate Optimization}}}
\author{{Liang~Jin,~Xiaodong~Xu,~\IEEEmembership{Senior~Member,~IEEE,}~Shujun~Han,~\IEEEmembership{Member,~IEEE,} 
		
		Xiaoyu~Chi,~Ping~Zhang,~\IEEEmembership{Fellow,~IEEE},~and~Chau~Yuen,~\IEEEmembership{Fellow,~IEEE}} 
\thanks{This work is supported in part by the National Key R$\&$D Program of China No. 2020YFB1806905; in part by the National Natural Science Foundation of China No. 62201079; and in part by MOE (Ministry of Education, Singapore), under MOE Tier 2 Award number T2EP50124-0032. {\textit {(Corresponding author: Xiaodong Xu)}}.}

\thanks{Liang Jin, Xiaoyu Chi and  Ping Zhang are with the State Key Laboratory of Networking and Switching Technology, Beijing University of Posts and Telecommunications, Beijing, 100876, China (e-mail: jinliang@bupt.edu.cn;  xiaoyu\_chi@bupt.edu.cn; pzhang@bupt.edu.cn).}

\thanks{Xiaodong Xu is with the State Key Laboratory of Networking and Switching Technology, Beijing University of Posts and Telecommunications, Beijing, 100876, China, and also with the Department of Broadband Communication, Peng Cheng Laboratory, Shenzhen 518066, Guangdong, China (e-mail: xuxiaodong@bupt.edu.cn; ).}

\thanks{Shujun Han is with the National Engineering Research Center for Mobile
	Network Technologies, Beijing University of Posts and Telecommunications,
	Beijing, 100876, China (e-mail: hanshujun@bupt.edu.cn).}

\thanks{Chau Yuen is with the School of Electrical and Electronics Engineering, Nanyang Technological University, Singapore 639798 (e-mail: chau.yuen@ntu.edu.sg).}

}

\markboth{IEEE latex,~Vol.~, No.~, July~2026}%
{IEEE latex,~Vol.~, No.~, July~2026}

\maketitle

\begin{abstract}
	Spatiotemporal channel coding (STCC) can improve the achievable rate over traditional temporal channel coding (TCC) by leveraging spatial degrees of freedom to extend the codeword length. Although several information-theoretic foundations on STCC have been established, the investigation of transmission schemes from a communication-theoretic perspective remains in its early stages.
	{	This paper proposes a multi-stream over multi-subchannel STCC (STCC-MSC) under full channel state information assumption and optimizes its achievable rate in the finite blocklength regime.} We first formulate the transmission architecture of STCC-MSC in a point-to-point MIMO system, which introduces a stream-subchannel matching mechanism. We then maximize the achievable rate of STCC-MSC by jointly optimizing the subchannel assignment and power allocation strategies, which is formulated as a mixed-integer-nonlinear-programming problem. Next, a  penalized alternating convex approximation (PACA) algorithm is proposed to  solve this problem. Subsequently, we extend the point-to-point STCC-MSC designs to the more general multi-user MIMO systems, including both uplink and downlink scenarios. Finally, simulation results indicate that the PACA algorithm achieves a 9.85\% rate improvement over the benchmark algorithm within the STCC-MSC scheme. Furthermore, the joint STCC-MSC-PACA scheme improves the achievable rate by 28.68\% over TCC scheme.
	
	\vspace{1em}
{	Accepted for publication in \textit{IEEE Trans. on Wireless Commun.}. The simulation code for reproduction is available at: \url{https://github.com/L-Jin-bupt}}  
\end{abstract}


\section{Introduction}
\IEEEPARstart{T}{o} {satisfy the tightening latency requirements in future applications, sixth-generation (6G) mobile communication systems are expected to achieve transmission delay on the order of microseconds \cite{expre,6G1,6G2}. However, due to the “impossible triangle” of blocklength-reliability-rate in communication systems, finite-blocklength communications suffer from significant rate loss because of the short length of each codeword \cite{turbo, Yury, durisi}.  In an effort to cope with the “impossible triangle”, researchers explore the space–blocklength exchangeability in multiple-input multiple-output (MIMO) systems, giving rise to the concept of spatiotemporal channel coding (STCC) \cite{ex}.}

{In traditional temporal channel coding (TCC), independent information streams are encoded in the time domain and transmitted over each subchannel, where the codeword length of each stream equals the number of channel uses. In contrast, STCC jointly encodes information streams across both time and multiple subchannels, effectively trading spatial degrees of freedom (DoF) for increased codeword length per stream \cite{CL}. As a result, despite supporting fewer simultaneously transmittable streams, STCC has been proved to achieve higher rate than TCC with identical time-frequency-space resources in the finite blocklength regime\cite{CL,FY,TWC,TVT}.} {It is worth noting that STCC is fundamentally different from conventional space-time coding (STC) \cite{STC1,STC2}. STC jointly designs transmitted symbols across the spatial and temporal dimensions to improve diversity and multiplexing performance. In contrast, STCC jointly maps channel-coded information onto spatial subchannels over multiple channel uses under the finite-blocklength information-theoretic framework to improve the achievable rate in the finite blocklength regime.}

Information-theoretic research on STCC dates back to 2010 and is now receiving growing attention from the communication community, while still remaining in its early stages. To the best of the authors’ knowledge, the earliest information-theoretic study related to STCC in the finite blocklength regime dates back to Polyanskiy’s dissertation \cite[Ch. 4.5]{FB2}, which analyzes the achievable rate of parallel Gaussian channels. However, he does not explicitly describe this as joint coding over time and space or distinguish it from traditional TCC. The term STCC is formally introduced in 2022 by You \textit{et al.}, who propose a spatiotemporal polar coding scheme based on product code and {show its superiority over TCC in block error rate \cite{polar}.} {Recently, more works on spatiotemporal polar coding have been developed\cite{n1,n2,n3}.} Moreover, the findings in \cite{polar} have been validated by experiments in a cell-free mmWave communication system \cite{expre}. Although \cite{polar} and \cite{expre} show promising results for STCC,  the underlying information-theoretic principles have not yet been fully explored.

{Building upon Polyanskiy's seminal non-asymptotic work in the finite-blocklength regime \cite[Ch. 4.5]{FB2}, some information-theoretic studies on STCC under various channel state information (CSI) assumptions in point-to-point (P2P) MIMO have emerged.} The achievable and converse bounds of STCC under different CSI conditions are investigated in \cite{YW1}, and no closed-form expressions are provided. Based on the results in \cite{YW1}, the average rate over Rayleigh fading channels is derived in \cite{CL} under the assumption that CSI is available at the receiver (CSI-R), while \cite{FY} investigates the achievable rate for a given channel realization under the same CSI assumption.  When CSI is available at both the transmitter (CSI-T) and the receiver, i.e., under full CSI, the transmitter can perform power allocation to better adapt to the channel and thereby achieve a higher rate. To this end, the classical water-filling method \cite[Ch. 4.5]{FB2} and its modified version \cite{np} have been adopted.\footnote{From an information-theoretic perspective, the parallel Gaussian channel effectively models a MIMO system with full CSI, as singular value decomposition can decompose the fading channel into $N_S$ independent Gaussian subchannels, where $N_S$ is the spatial DoF \cite[Ch. 7]{Tse}.}  Additionally, from the perspective of electromagnetic information theory, we provide a physical justification for the STCC design in \cite{TWC}. All the aforementioned works \cite{CL,FY,TWC,FB2,YW1,np} consider STCC with a single transmitted stream, i.e., all subchannels participate in encoding only one stream. { Motivated by this, our prior work \cite{TVT} extends the STCC framework to the multi-stream scenario under the CSI-R assumption and focuses on the achievable rate analysis with linear receivers.} A clearer summary of the above literature is provided in Table \ref{tab:summary}.

{Unlike the CSI-R setting in \cite{TVT}, this paper proposes a multi-stream over multi-subchannel STCC (STCC-MSC) transmission scheme that considers both transmitter-side design and the corresponding receiver processing under the full-CSI assumption, and develops the corresponding achievable rate optimization framework.} Our main contributions are summarized as follows:

\begin{table*}[t]
	\centering
	\caption{Existing Works on Achievable Rate of STCC}
	{
		\begin{tabular}{lccccccc}
			\toprule
			\textbf{Reference}  & \textbf{Closed-form rate expression}& \textbf{CSI-R} & \textbf{CSI-T} &\textbf{Multi-stream}&\textbf{Multi-user} \\
			\midrule
			\cite{YW1} & \xmark & \cmark & \cmark & \xmark & \xmark  \\
			\cite{FY,CL} & \cmark & \cmark & \xmark & \xmark & \xmark  \\
			\cite{np,FB2, TWC} & \cmark & \cmark & \cmark & \xmark& \xmark   \\
			\cite{TVT} & \cmark & \cmark & \xmark & \cmark & \xmark  \\
			This paper & \cmark & \cmark & \cmark & \cmark & \cmark  \\
			\bottomrule
		\end{tabular}
	}
	\label{tab:summary}
\end{table*}

\begin{table}[t]
	\caption{Notations Used in This Paper}
	\label{not}
	\centering
	\begin{tabular}{ll}
		\toprule
		\textbf{Notations} &\textbf{Meanings}\\
		\midrule
		$N_t,N_r$ & Number of antennas in P2P MIMO\\
		$d,i,k$ & Index of streams/subchannels/users\\
		$N_S, N_{S,k}$ & {Spatial DoF}\\
		$D$ &{Number of maximum streams}\\
		{$n$} & {Number of channel uses}\\
		{${K}$} & {Number of users}\\
		$N_B,N_k$ & Number of antennas in multi-user MIMO\\
		$\overline{\mathbf{X}},\mathbf{X}, {\mathbf{x}}_d,\overline{\mathbf{x}}_k$ & {Transmitted codewords}\\
		$\overline{\mathbf{Y}},\mathbf{Y}, {\mathbf{y}}_d,\overline{\mathbf{y}}_k$ & {Received codewords} \\		
		$\mathbf{H}, \overline{\mathbf{H}},\mathbf{H}_k,\overline{\mathbf{H}}_k$ & {Channel matrices}\\ 
		{$\overline{\mathbf{W}},\mathbf{W}, \mathbf{w}_k$} & {AWGN matrices/vector}\\
		$\mathbf{U}$ & {Left singular vectors in SVD}\\
		$\mathbf{V}$ & {Right singular vectors in SVD}\\
		$ \mathbf{\Lambda}$ & {Singular value matrices}\\
		$\lambda_{i}$ & Singular value of the $i$-th subchannel\\
		$\xi_{i}$ & SNR of the $i$-th subchannel\\
		$\bm{\xi_{d}}$ & SNRs of the  subchannels of the $d$-th stream\\
		$\left(\cdot\right)^H,\left(\cdot\right)^T$ & {Hermitian transpose/transpose}\\
		{$\mathbf{M}(i,:)$} & {The \textit{i}-th row of matrix $\mathbf{M}$}\\
		{$\mathbf{M}(:,d)$} & {The \textit{d}-th column of matrix $\mathbf{M}$}\\
		$\mathbf{M}(i,d)$ & Element in the $i$-th row and $d$-th column of $\mathbf{M}$\\
		$\mathbf{m}(b:c)$ &  Subvector of $\mathbf{m}$ from the $b$-th to  $c$-th element\\
		$\varepsilon$ & {Block error probability}\\
		$\tilde{i}(x;y)$ & Information density \\
		{$\mathbf{P},\mathbf{p},\mathbf{p}_k$} & {Transmit power matrix/vectors}\\
		$\sigma^2,\sigma_k^2$& Noise powers\\
		${\rm log}$ & {Binary logarithm} \\
		$\mathcal{CN}$ & {Complex normal distribution}\\
		{$\mathbf{S},\mathbf{S}_k$} & {Subchannel assignment matrices}\\
		$V,V_i$ & {Channel dispersions}\\
		$R_T, R_{ST},R_d,R_k$ & {Achievable rates} \\
		$\mathcal{N_S}, \mathcal{D}, \mathcal{K}$ & {Set of subchannels/streams/users }\\
		$\mathcal{N}_{S,k}, \mathcal{D}_k$ & {Set of subchannels/streams of user $k$}\\
		$Q^{-1}$ & {Inverse Q function}\\
		$\mathbb{R}^{r \times c},\mathbb{C}^{r \times c}$ & {Real/complex matrix with dimension $r \times c$ }\\
		$\mathbf{G},\mathbf{Q}$ & {Auxiliary variables}\\
		$\mathbf{I}_{r}$ & {Identity matrix with dimension $r \times r$ }\\
		$\mathbf{0}_{r \times 1}$ & {Zero vector with dimension $r \times 1$ }\\
		$\epsilon,\epsilon_O$ & Precision threshold\\
		$P,P_D,P_k$ & Maximum transmit powers\\
		{$||\cdot||_2$, $||\cdot||_0$} & {Euclidean norm/$\ell_0$-norm}\\
		$\odot$ & Hadamard product\\
	{	$\{0,1\}^{N_S \times D}$} & {Set of $N_S \times D$-dimension binary matrices}\\
		\bottomrule
	\end{tabular}
\end{table}

\begin{enumerate}
	\item We formulate the transmission architecture of STCC-MSC in a P2P MIMO system. Compared with the conventional TCC scheme, the proposed architecture of STCC-MSC introduces a stream–subchannel matching mechanism, enabling each stream to be transmitted over multiple subchannels and encoded with a longer codeword length.
	
	\item Since the STCC-MSC entails a co-design of stream–subchannel matching and AWGN channel coding, we aim to maximize the achievable rate by jointly optimizing the subchannel assignment and power allocation strategies. To solve this mixed-integer nonlinear programming (MINLP) problem efficiently, we propose a penalized alternating convex approximation (PACA) algorithm.

	\item We extend the P2P STCC-MSC designs to the more general multi-user scenarios. Specifically, by applying the block diagonalization (BD) technique to eliminate inter-user interference, the downlink sum-rate maximization problem is reformulated as a $K$-fold higher-dimensional extension of the P2P case. In contrast, the uplink sum-rate maximization problem decomposes into $K$ parallel and independent sub-problems, each equivalent to a rate maximization problem in the P2P setting. Both uplink and downlink sum-rate maximization problems can be efficiently solved using the proposed PACA algorithm.
	
	\item Simulation results indicate that the PACA algorithm achieves a 9.85\% rate improvement over the benchmark algorithm within the STCC-MSC scheme. Furthermore, the joint STCC-MSC-PACA scheme improves the rate by 28.68\% over the TCC scheme. 
\end{enumerate}

The rest of this paper is organized as follows. In Section \uppercase\expandafter{\romannumeral2}, we introduce some preliminaries, including the interference management in MIMO, as well as the motivation and the design principle of STCC. In Section \uppercase\expandafter{\romannumeral3}, we introduce the  transmission architecture of STCC-MSC and formulate the rate maximization problem. In Section \uppercase\expandafter{\romannumeral4}, the PACA algorithm is introduced. In Section \uppercase\expandafter{\romannumeral5}, we extend the designs to the multi-user MIMO scenario. Simulation results are presented in Section \uppercase\expandafter{\romannumeral6}. Conclusions and discussions are given in Section \uppercase\expandafter{\romannumeral7}. In this paper, we use the boldface upper case letters represent matrices and the boldface lower case letters represent
vectors. Notations used in this paper are listed in Table \ref{not}.

\section{Preliminaries}
This section presents the motivation and the design principle of STCC and clarifies how it differs from TCC. We first briefly review the application of singular value decomposition (SVD) in mitigating spatial multiplexing-induced interference in P2P MIMO systems. Then we present the design principle of traditional TCC and introduce the motivation and design principles of STCC. 

\begin{figure*}[t!]
	\centerline{\includegraphics[width=1\linewidth]{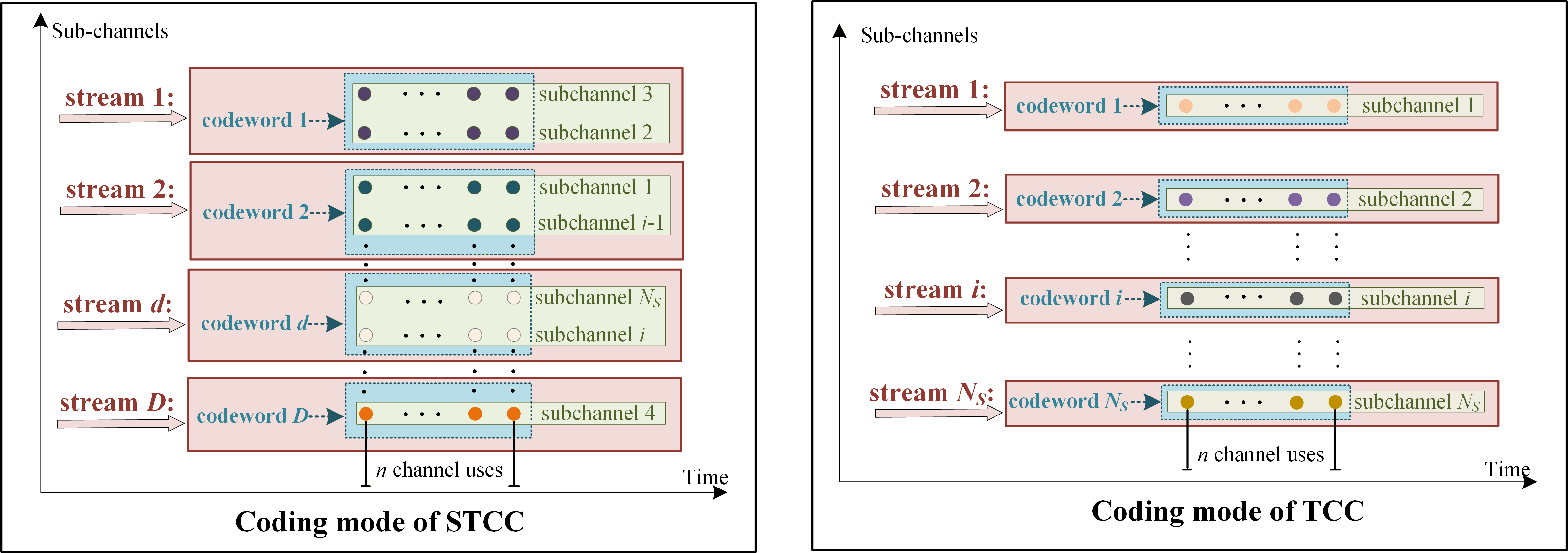}}
\caption{{Coding modes of STCC and TCC for MIMO with multiple streams.}}
	\label{codingmodes}
\end{figure*}

\subsection{Transforming Fading Channel into Parallel Gaussian Channel by Using SVD}
{Consider a P2P MIMO system, where the transmitter and the receiver are equipped with $N_t$ and $N_r$ antennas, respectively.} Since the codewords require to be transmitted within a short blocklength, it is reasonable to consider a quasi-static flat fading MIMO channel, where the channel matrix $\mathbf{H}\in \mathbb{C}^{N_r\times N_t}$ remains invariant throughout the duration of each codeword \cite{FY}. $\overline{\mathbf{X}}\in \mathbb{C}^{N_t \times n}$ and $\overline{\mathbf{Y}}\in \mathbb{C}^{N_r \times n}$ denote the transmitted codeword and received codeword at antennas, {where $n$ is the number of channel uses, i.e., the blocklength.} In a communication system, the number of antennas characterizes its spatial dimension, while $n$ serves as a measure of the temporal dimension. Let $\mathbf{M}(:,m)$ denote the $m$-th column of matrix $\mathbf{M}$. The transmitted and received symbol vectors at the $m$-th channel use are denoted by $\overline{\mathbf{X}}(:,m) \in \mathbb{C}^{N_t \times 1}$ and $\overline{\mathbf{Y}}(:,m) \in \mathbb{C}^{N_r \times 1}$, respectively, and they satisfy
\begin{equation}
	\overline{\mathbf{Y}}(:,m) = \mathbf{H} \, \overline{\mathbf{X}}(:,m) + \overline{\mathbf{W}}(:,m), \label{n1}
\end{equation}
{where $\overline{\mathbf{W}}(:,m) \sim \mathcal{CN}(0, \sigma^2\mathbf{I}_{N_{r}})$ denotes the additive white Gaussian noise (AWGN) vector, and all entries of the noise matrix $\overline{\mathbf{W}}$ are independent and identically distributed complex Gaussian random variables.}

{Let $N_S=\text{rank}(\mathbf{H})$ denote the spatial DoF of the channel. The left singular matrix \( \mathbf{U} \in \mathbb{C}^{N_r \times N_{S}} \) and the right singular matrix \( \mathbf{V} \in \mathbb{C}^{N_t \times N_{S}} \) consist of the first \( N_{S} \) columns of the singular matrices obtained from the SVD of \( \mathbf{H}\). $\bm{\Lambda} \in \mathbb{R}^{N_S \times N_S}$ is a diagonal matrix, {where the diagonal elements \( \lambda_1 \geq \lambda_2 \geq \cdots \geq \lambda_{{N_S}} \) are the ordered non-zero singular values of the matrix \( \mathbf{H}\mathbf{H}^H \) if \( N_r \ge N_t \); otherwise, they are those of \( \mathbf{H}^H\mathbf{H} \).} Let $\mathbf{p}=[p_1,p_2 ...,p_{N_S}]^T$, where $p_i$ is the power allocated to the $i$-th subchannel, and $\mathbf{P}={\rm diag}\left(\mathbf{p}\right) \in \mathbb{R}^{N_S \times N_S}$. By letting 
\begin{equation} \label{SVD1}
	\begin{aligned}
		{\mathbf{X}}(:,m)&=\mathbf{P}\mathbf{V}^H\overline{\mathbf{X}}(:,m),\\
		{\mathbf{Y}}(:,m)&=\mathbf{U}^H\overline{\mathbf{Y}}(:,m),\\
		{\mathbf{W}}(:,m)&=\mathbf{U}^H\overline{\mathbf{W}}(:,m),\\
	\end{aligned}
\end{equation}
the fading channel (\ref{n1}) is transformed into following parallel Gaussian channel:
\begin{equation} \label{SVD2}
	{\mathbf{Y}}(:,m)=\mathbf{P}\boldsymbol{\Lambda}{\mathbf{X}}(:,m)+{\mathbf{W}}(:,m), 
\end{equation} 
where ${\mathbf{W}}(:,m)\sim \mathcal{CN}(0, \sigma^2\mathbf{I}_{N_S})$ has the same distribution as  $\overline{\mathbf{W}}(:,m)$ since $\mathbf{U}$ is a unitary matrix.  This transformation enables the application of information-theoretic results to compute the rate of MIMO systems, which has guided the development of communication theory over the past decades, see \cite[Ch.7]{Tse} for more details.}

\subsection{Design Principles of TCC and STCC}
TCC and STCC operate on the same channel model in (\ref{SVD2}) but differ in their coding modes, as illustrated in Fig. \ref{codingmodes}. On the right of Fig. \ref{codingmodes}, we depict the coding mode of TCC, where independent information streams are encoded temporally and transmitted on each subchannel. Therefore, the TCC scheme supports the concurrent transmission of up to $N_S$ independent streams, with the blocklength of each codeword to be $n$. {Based on the finite-blocklength normal approximation established by Polyanskiy et al. \cite{Yury}, the closed-form achievable-rate expression for the TCC scheme is given by \cite{FY}}
\begin{equation}
	{{R}_{T}}=\sum\limits_{i=1}^{{{N}_{s}}}{\left( \log \left( 1+\frac{{{p}_{i}}{{\lambda }_{i}}}{\sigma ^{2}} \right)-{{\sqrt{{V}}_{i}}}\frac{{{Q}^{-1}}\left( \varepsilon  \right)}{\sqrt{n}}\log e \right)}, \label{RT}
\end{equation}
{where $\varepsilon$ is the block error probability, $Q^{-1}(\cdot)$ denotes the inverse Gaussian Q-function,}  $p_i$ is the transmit power allocated to the $i$-th subchannel and
\begin{equation}
	{{{V}}_{i}}=1-\frac{1}{{{\left( 1+\frac{{{p}_{i}}{{\lambda }_{i}}}{\sigma ^{2}} \right)}^{2}}}, \label{Vi}
\end{equation}
which is the channel dispersion. {Strictly speaking, the exact normal approximation of finite-blocklength achievable-rate contains an $\mathcal{O}\left(\frac{\log n}{n}\right)$ remainder term \cite{Yury}. Since this term is characterized only in terms of its asymptotic order and generally admits no closed-form expression, it is commonly omitted in analytical derivations and optimization \cite{mono,r1,r2}. Moreover, since the blocklength $n$ is fixed in our optimization problem, even the commonly used correction term 
	$\frac{{\rm log}n}{2n}$ would only introduce a constant offset to the achievable rate and would not affect the optimization framework or the resulting resource allocation. 
	Therefore, following the common practice in \cite{mono,r1,r2}, we neglect the remainder term in our derivations and optimization to obtain a concise and tractable achievable-rate expression.}

The TCC scheme often suffers from significant rate loss due to finite blocklength, as we can observe in (\ref{RT}). Therefore, researchers propose STCC, trying to increase the rate of MIMO systems by trading the spatial DoF for increasing the codeword length of each stream \cite{ex, polar}.
As depicted on the left of Fig. \ref{codingmodes}, multiple subchannels are used to transmit a stream cooperatively in the STCC scheme. Consequently, the STCC scheme trades spatial DoF for an increased codeword length for each codeword, which in turn reduces the maximum allowable number of streams $D$, satisfying $D \le N_S$. {It should be noted that in the Fig. 1, the order of subchannels is shuffled to illustrate that the STCC streams can occupy non-adjacent subchannels, rather than being restricted to adjacent ones.}

{According to the finite-blocklength normal approximation in \cite[Ch. 4.5]{FB2}, the achievable rate of STCC for $D=1$ is expressed as}
\begin{equation} \label{D=1}
	R_{ST|D=1}=\sum\limits_{i=1}^{{{N}_{s}}}{ \log \left( 1+\frac{{{p}_{i}}{{\lambda }_{i}}}{\sigma ^{2}} \right)}-{{\sqrt{\sum\limits_{i=1}^{{{N}_{s}}}{{V}}_{i}}}}\frac{{{Q}^{-1}}\left( \varepsilon  \right)}{\sqrt{n}}\log e,
\end{equation}
where ${V}_i$ is given by (\ref{Vi}) and $\mathbf{p}=[p_1,p_2 ...,p_{N_S}]^T$ should be determined by the power allocation strategy. {As explained following Eq. (\ref{RT}), the remainder term is also omitted in Eq. (\ref{D=1})  to obtain a concise and tractable achievable-rate expression while focusing on the optimization.} It has been proven that $R_{\mathrm{ST}\mid D=1} > R_{\mathrm{T}}$ when transmit power is equally allocated under CSI-R \cite{CL,FY}. However, as far as we know, no results have been reported for the case of multiple streams in STCC when full CSI is available. In such a case, the transmission architecture should be designed and the signal processing method, which is ultimately reflected in the form of subchannel assignment and power allocation, should be determined. This paper aims to address the above issues and explore the potential superiority of STCC over TCC.

\begin{figure*}[t!]
	\centerline{\includegraphics[width=1\linewidth]{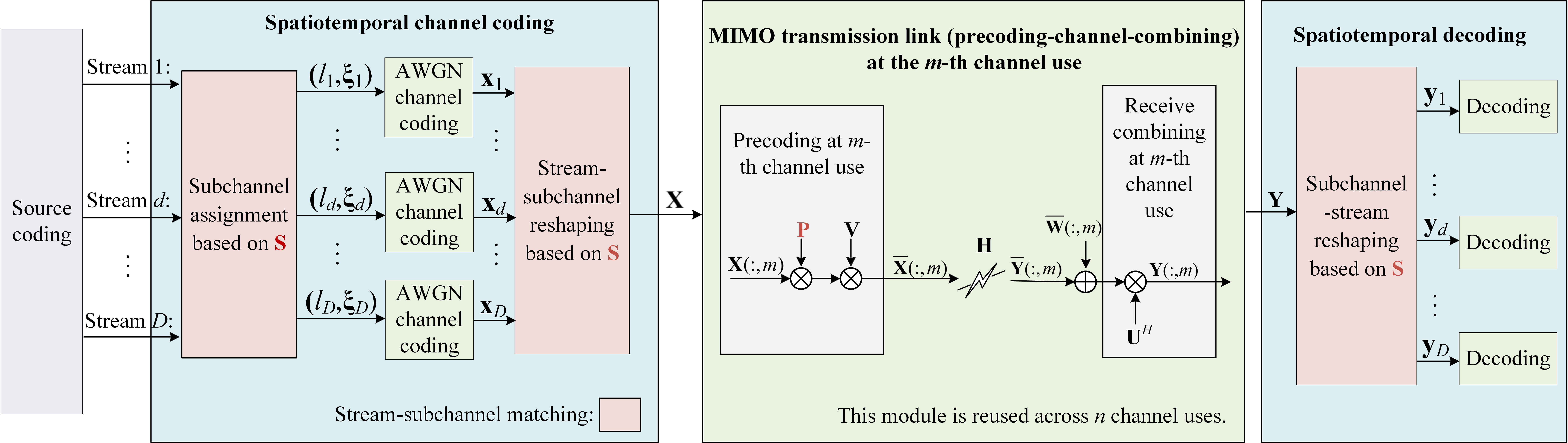}}
	\caption{{Equivalent baseband transmission architecture of STCC-MSC.}}
	\label{system}
\end{figure*}

\section{Transmission Architecture and Rate Maximization Problem Formulation}

In this section, we first build the transmission architecture of STCC-MSC in the P2P MIMO system and then formulate the corresponding rate maximization problem. An effective algorithm for solving this problem is proposed in Section \uppercase\expandafter{\romannumeral4}, and the multi-user scenario is investigated in Section \uppercase\expandafter{\romannumeral5}.

\subsection{Transmission Architecture of STCC-MSC}

{
We still use the system model described in Section \uppercase\expandafter{\romannumeral2}, where the transmitter and receiver are equipped with $N_t$ and $N_r$ antennas, respectively. Same as \cite[Ch. 7]{Tse}, we analyze the equivalent baseband channel for brevity and omit the radio frequency part in this paper.  The equivalent baseband fading channel is $\mathbf{H}$ with rank$(\mathbf{H})=N_S$. The number of allowable streams is $D$, satisfying $D \le N_S$. The number of channel uses is $n$. Since the channel matrix $\mathbf{H}$ remains invariant during the transmission, the precoding and combining strategy can be reused over all $n$ channel uses. With loss of generality, we focus on the $m$-th channel use for analysis and illustration. At the $m$-th channel use, the parallel Gaussian channel in (\ref{SVD2}) is obtained after applying precoding by left-multiplying $\mathbf{P}$ and $\mathbf{V}$, and receive combining by left-multiplying $\mathbf{U}^H$. As described in Section \uppercase\expandafter{\romannumeral2}-A, $\mathbf{V}$ and $\mathbf{U}$ are determined by $\mathbf{H}$, whereas $\mathbf{P}$ needs to be designed.
Let $\mathcal{D}=\left\{1,2,...,D\right\}$ and $\mathcal{N}_S=\left\{1,2,...,N_S\right\}$ denote the set of streams and subchannels, respectively.  $\bm{\xi}=[\xi_1,\xi_2,...,\xi_{N_S}]^T$ denotes the SNRs of the subchannels, with $\xi_i = \frac{p_i\lambda_{i}}{\sigma^2}, \forall i \in \mathcal{N_S}$. }

{
Compared with the traditional TCC architecture, which involves only AWGN channel coding~\cite[Fig.~7.2]{Tse}, the proposed  architecture of STCC-MSC introduces a co-design of AWGN coding and stream–subchannel matching, enabling each stream to be transmitted over multiple subchannels, as illustrated in Fig.~\ref{system}. {The stream–subchannel matching is determined by the assignment matrix $\mathbf{S} \in \{0,1\}^{N_S \times D}$, where $\{0,1\}^{N_S \times D}$ denotes the set of $N_S \times D$-dimension binary matrices.} Let $\mathbf{S}(i,d)$ denote the element in the $i$-th row and $d$-th column of $\mathbf{S}$.  $\mathbf{S}(i,d)=1$ indicates that the $i$-th subchannel is assigned to the $d$-th stream, and $\mathbf{S}(i,d)=0$ otherwise. The matching is logically divided into two stages, which are executed before and after AWGN channel coding, respectively.}

Firstly, STCC-MSC assigns the $N_S$ subchannels to $D$ streams based on $\mathbf{S}$, resulting in the sets $\left\{\bm{\xi}_d\right\}_{d=1}^D$ and $\left\{l_d\right\}_{d=1}^D$. Here, $l_d$ denotes the number of nonzero entries in the $d$-th column of $\mathbf{S}$, indicating the number of subchannels assigned to the $d$-th stream. {Consequently, every stream in STCC is transmitted over the same $n$ channel uses, while the codeword length of stream $d$ is $l_dn$ transmitted symbols.}
 The vector $\bm{\xi}_d \in \mathbb{R}^{l_d \times 1}$ contains the SNRs of the subchannels assigned to the $d$-th stream and is given by
\begin{equation}
	\bm{\xi}_d = \left( \mathbf{S}(:,d) \odot \boldsymbol{\xi} \right)_{+}, \quad \forall d \in \mathcal{D},
\end{equation}
where $\mathbf{S}(:,d)$ denotes the $d$-th column of $\mathbf{S}$, $\odot$ denotes the Hadamard (element-wise) product, and $(\mathbf{m})_{+}$ denotes the operation that removes the zero entries from the vector $\mathbf{m}$. Next, based on the obtained $l_d$ and $\bm{\xi}_d$, the codeword of the $d$-th stream after AWGN channel coding is denoted by $\mathbf{x}_d \in \mathbb{C}^{1 \times l_d n}$. Finally, the overall spatiotemporal codeword $\mathbf{X} \in \mathbb{C}^{N_S \times n}$ is constructed according to the stream-subchannel reshaping based on $\mathbf{S}$. Let $\mathbf{M}(i,:)$ denote the $i$-th row of a matrix $\mathbf{M}$, and let $\mathbf{m}(b:c)$ denote the subvector of $\mathbf{m}$ from the $b$-th to the $c$-th element. The detailed matching process is presented in \textbf{Algorithm~\ref{Al1}}. An illustrative example of STCC is shown in Fig.~\ref{s2p}, where 3 streams are mapped to 6 subchannels. According to the given assignment matrix $\mathbf{S}$, the values of $l_1$, $l_2$, and $l_3$ are 1, 2, and 3, respectively. { Fig. 3 shows only one illustrative stream-subchannel assignment; in general, $\mathbf{S}$ satisfies  $\sum\limits_{d=1}^D\mathbf{S}(i,d)\le1, \forall i \in \mathcal{N_S}$, and therefore a subchannel is not necessarily assigned to any stream.} {Moreover, The ordering of $\mathbf{x}_1, \mathbf{x}_2, \mathbf{x}_3$
	does not affect the achievable rate, which is determined by the stream-subchannel assignment.}

At the receiver side, the received signal matrix $\mathbf{Y}$ is converted into a set of stream-specific observations $\left\{\mathbf{y}_d\right\}_{d=1}^D$ through a subchannel-stream reshaping, which essentially reverses the process described in \textbf{Algorithm~\ref{Al1}}. {Since the reshaping process is information-lossless and each stream operates as an independent system comprising $l_d$ subchannels, the total information density between the transmitted and received signals equals the sum of stream-level information densities, satisfying $\tilde{i}(\mathbf{X}; \mathbf{Y}) = \sum_{d=1}^{D} \tilde{i}(\mathbf{x}_d; \mathbf{y}_d)$.} Applying (\ref{D=1}), the achievable rate of the $d$-th stream is given by
\begin{equation} \label{Rd}
	R_{d}=\sum\limits_{i=1}^{{{N}_{s}}}\mathbf{S}(i,d){ \log \left( 1+\frac{{{p}_{i}}{{\lambda }_{i}}}{\sigma^{2}} \right)}-{a_{n,\varepsilon}}{{\sqrt{\sum\limits_{i=1}^{{{N}_{s}}}{\mathbf{S}(i,d){V}}_{i}}}},
\end{equation}
where ${a_{n,\varepsilon}}=\frac{{{Q}^{-1}}\left( \varepsilon  \right)}{\sqrt{n}}\log e$ and ${V}_{i}$ is given by (\ref{Vi}).   The rate of the P2P MIMO system using STCC-MSC is the sum of the rates of all streams as
\begin{equation} \label{RST}
	R_{ST}=\sum\limits_{d=1}^{D}R_d.
\end{equation}

\begin{algorithm}[t]
	\caption{Stream-subchannel reshaping in STCC-MSC}
	\label{Al1} 
	\begin{algorithmic}[1]
		\Statex \textbf{Input}: $\left\{{\mathbf{x}}_d\right\}_{d=1}^D$, $\mathbf{S}$. 
		\vspace{0.1cm}
		\Statex \textbf{Output:} $\mathbf{X}$. 
		\State Initialize $\mathbf{X}=\mathbf{0}_{N_S \times n}$, $d=1$, $i = 1$, and $k = 0$.
		\Repeat
		\Repeat
		\If{$\mathbf{S}(i,d) = 1$}
		\State $\mathbf{X}(i,:) = 			{\mathbf{x}}_d(kn+1:(k+1)n)$.
		\State Update $k = k + 1$.
		\EndIf
		\State Update $i = i + 1$.
		\Until{$i = N_S$}
		\State Reset $k=0$ and update $d = d + 1$.
		\Until{$d = D$}
		\State \Return $\mathbf{X}$.
	\end{algorithmic}
\end{algorithm}

\begin{figure}[t]
	\centerline{\includegraphics[width=1\linewidth]{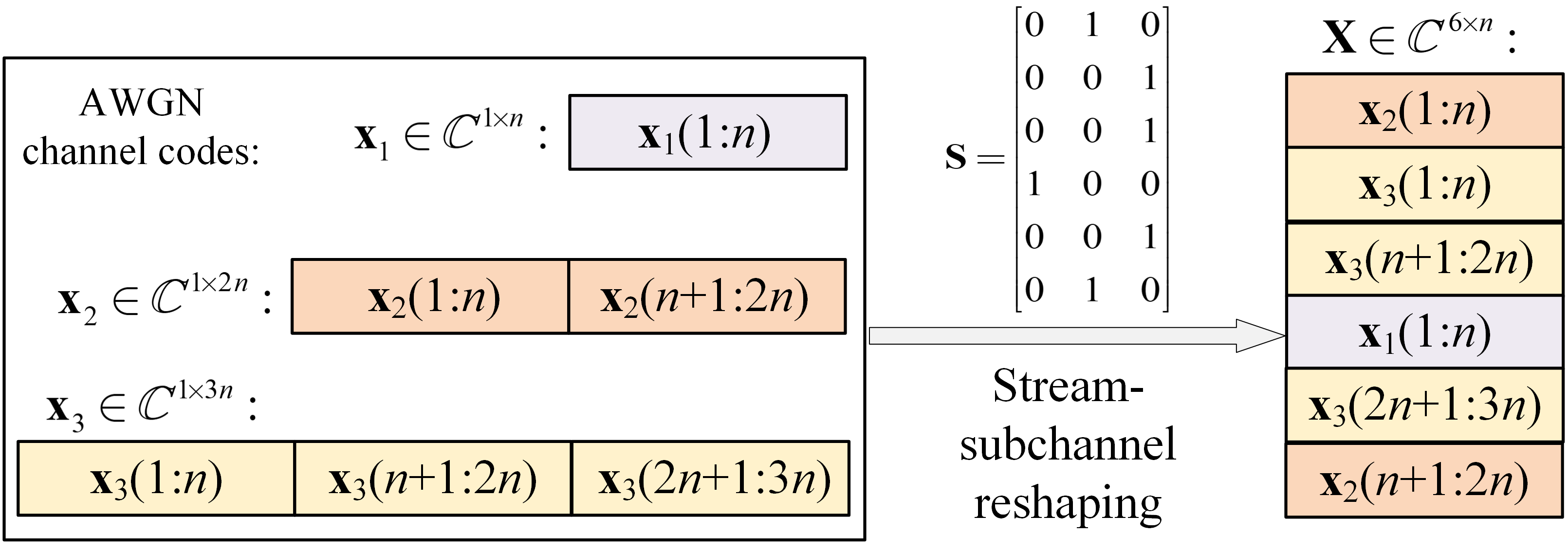}}
	\caption{{A case illustration of STCC-MSC.}}
	\label{s2p}
\end{figure}

\subsection{Rate Maximization Problem Formulation}
 Let $\gamma_i = \frac{\lambda_{i}}{\sigma^2}, \forall i \in \mathcal{N_S}$.   By referring to (\ref{Rd}) and (\ref{RST}), the rate maximization problem for STCC can be equivalently formulated by minimizing $-R_{ST}$ as
\begin{subequations} \label{P1}
	\begin{align}
		 \underset{\mathbf{p},\mathbf{S}}{\mathop{\min }}\,&\sum\limits_{d=1}^{D}{\left( {a_{n,\varepsilon}}\sqrt{\sum\limits_{i=1}^{{{N}_{S}}}{\mathbf{S}(i,d)V_i(p_i)}}-\sum\limits_{i=1}^{{{N}_{S}}}{\mathbf{S}(i,d)\log \left( 1+{{p}_{i}}{{\gamma }_{i}} \right)} \right)} \label{obj} \\ 
		 \text{  s}\text{.t}\text{. }&\sum\limits_{i=1}^{{{N}_{S}}}{{{p}_{i}}\le P},\text{    } \label{con1}\\ 
		& {{p}_{i}}\ge 0,\text{ }\forall i\in \mathcal{N}_S, \label{con2}\\ 
		& \mathbf{S}(i,d)=\{0,1\},\text{ }\forall i\in \mathcal{N}_S,\text{ } \forall d\in \mathcal{D}, \label{con3}\\ 
		& \sum\limits_{d=1}^{D}\mathbf{S}(i,d) \le 1, \text{ }\forall i\in \mathcal{N}_S. \label{con4}
	\end{align}
\end{subequations}
Constraint (\ref{con1}) caps the transmit power at $P$ and constraint (\ref{con2}) requires non-negative transmit power for each subchannel. (\ref{con3}) are binary constraints for the elements of the subchannel assignment matrix $\mathbf{S}$. {Constraints (\ref{con4}) mean that each subchannel can be assigned to at most one stream.}

{Problem (\ref{P1}) is a non-convex MINLP problem because of the coupling of the optimization variables} in the objective function (\ref{obj}) and the binary constraints (\ref{con3}). {It is intractable to find the globally optimal solution to problem (\ref{P1}) in polynomial time.} Although some search methods, such as exhaustive search and the Branch-and-Bound algorithm, can theoretically find the global optimum, their prohibitively high complexity limits practical applicability \cite{ADMM}. Therefore, it is necessary to find a locally optimal solution with reasonable computational complexity.

{A problem similar to (\ref{P1}) is studied in \cite{PE},} where multiple sub-carriers are used to transmit a stream to each single-antenna user, in the multi-user downlink orthogonal frequency division multiple access (OFDMA) system. The authors in \cite{PE} use the big-M technique to decouple the variables, and then propose a successive convex approximation (SCA) algorithm to tackle the non-convexity in the objective function and the binary constraints.  {However, the convex approximation disrupts the structure of the big-M reformulation, yielding solutions that violate the original constraints and thereby degrading the algorithm performance.} Another commonly used technique to decouple the variables is alternating optimization, which optimizes some variables while fixing others during each iteration. {However, using alternating optimization to optimize $\mathbf{S}$ and $\mathbf{p}$ iteratively is not applicable in solving problem (\ref{P1}) because optimizing $\mathbf{S}$ with fixed $\mathbf{p}$ is a spurious optimization problem, as discussed in Appendix \ref{appa1}.} Therefore, in the next section, {we propose a penalized alternating convex approximation algorithm to efficiently obtain a locally optimal solution to problem (\ref{P1}).} Numerical results in Section \uppercase\expandafter{\romannumeral6} demonstrate that it outperforms the big-M convex approximation (BMCA) algorithm proposed in \cite{PE}.

\section{Proposed Penalized Alternating Convex Approximation Algorithm}
{We define a matrix $\mathbf{Q} \in \mathbb{R}^{N_S \times D}$, whose $d$-th column $\mathbf{Q}(:,d)=\mathbf{S}(:,d)\odot\mathbf{p}, \forall d \in \mathcal{D}$, where $\mathbf{M}(:,d)$ denotes the $d$-th column vector of matrix $\mathbf{M}$.} The original problem (\ref{P1}) can be equivalently rewritten as
\begin{subequations} \label{P2}
	\begin{align}
		{\rm P1}:\text{ } \underset{\mathbf{Q}}{\mathop{\min }}\text{ }\,&V(\mathbf{Q})-C(\mathbf{Q}) \label{p2obj}\\ 
		 \text{  s}\text{.t}\text{. }&\sum\limits_{d=1}^{D}{\sum\limits_{i=1}^{{{N}_{S}}}{\mathbf{Q}(i,d)}}\le P, \\ 
		& ||\mathbf{Q}(i,:)|{{|}_{0}}\le 1,\text{ }\forall i\in \mathcal{N}_S, \label{p2con2} \\ 
		& \mathbf{Q}(i,d)\ge 0,\text{ }\forall i\in \mathcal{N}_S,\text{ }\forall d\in \mathcal{D}. 
	\end{align}
\end{subequations}
where
\begin{equation}
	V(\mathbf{Q})={a_{n,\varepsilon}}\sum\limits_{d=1}^{D} \sqrt{\sum\limits_{i=1}^{{{N}_{S}}}{\left(1-\frac{1}{{{\left( 1+\mathbf{Q}(i,d){{\gamma }_{i}} \right)}^{2}}}\right)}},
\end{equation}
\begin{equation}
	C(\mathbf{Q})=\sum\limits_{d=1}^{D}\sum\limits_{i=1}^{{{N}_{S}}}{\log \left( 1+\mathbf{Q}(i,d){{\gamma }_{i}} \right)},
\end{equation}
$\mathbf{Q}(i,:)$ denotes the $i$-th row of the matrix $\mathbf{Q}$, and $||\mathbf{m}|{{|}_{0}}$ is the $\ell_0$-norm of vector $\mathbf{m}$, which is number of non-zero elements in vector $\mathbf{m}$. 

Problem (P1) is still non-convex because of the concave objective function (\ref{p2obj}) and the sparsity constraints (\ref{p2con2}). We first apply the penalty decomposition method to tackle constraints   (\ref{p2con2}). By introducing the auxiliary variable $\mathbf{G} \in \mathbb{R}^{N_S \times D}$ and the penalty term, we rewrite problem (P1) as \cite{PD}
\begin{subequations} \label{P3}
	\begin{align}
	{\rm P2}:\text{ } \underset{\mathbf{Q},\mathbf{G}}{\mathop{\min }}\text { }\,&V(\mathbf{Q})-C(\mathbf{Q})+\sum\limits_{i=1}^{N_S}{\rho||\mathbf{G}(i,:)-\mathbf{Q}(i,:)||_2^2} \label{p3obj}\\ 
		\text{  s}\text{.t}\text{. }&\sum\limits_{d=1}^{D}{\sum\limits_{i=1}^{{{N}_{S}}}{\mathbf{Q}(i,d)}}\le P, \\ 
		& \mathbf{Q}(i,d)\ge 0,\text{ }\forall i\in \mathcal{N}_S,\text{ }\forall d\in \mathcal{D},\\
		& ||\mathbf{G}(i,:)|{{|}_{0}}\le 1,\forall i\in \mathcal{N}_S, \label{p3con2} 
	\end{align}
\end{subequations}
where $\rho$ is the penalty factor and $||\mathbf{m}||_2$ is the $l_2$ norm of the vector $\mathbf{m}$. When $\rho$ gradually increases and becomes large enough, for the optimal solution of problem (P2), it holds that $\mathbf{G}=\mathbf{Q}$, so problem (P2) is equivalent to problem (P1). Although problem (P2) is still non-convex, we can apply the alternating optimization method to decouple it into two sub-problems: (1) optimizing $\mathbf{Q}$ with fixed $\mathbf{G}$ and (2) optimizing $\mathbf{G}$ with fixed $\mathbf{Q}$. By solving the two sub-problems iteratively, we can obtain a locally optimal solution of problem (P2). 

\textit{(1) Optimizing $\mathbf{Q}$ with fixed $\mathbf{G}$}: this sub-problem of (P2) is expressed as
\begin{subequations} \label{P3.1}
	\begin{align}
	{\rm P2.1}:\text{ }	\underset{\mathbf{Q}}{\mathop{\min }}\text { }\,&V(\mathbf{Q})-C(\mathbf{Q})+\sum\limits_{i=1}^{N_S}{\rho||\mathbf{G}(i,:)-\mathbf{Q}(i,:)||_2^2} \label{p31obj}\\ 
		\text{  s}\text{.t}\text{. }&\sum\limits_{d=1}^{D}{\sum\limits_{i=1}^{{{N}_{S}}}{\mathbf{Q}(i,d)}}\le P, \label{p31con1}\\ 
		& \mathbf{Q}(i,d)\ge 0,\text{ }\forall i\in \mathcal{N}_S,\forall d\in \mathcal{D} \label{p31con2}.
	\end{align}
\end{subequations}
Problem (P2.1) is still non-convex because of the concave function $V(\mathbf{Q})$ in the objective function (\ref{p31obj}). {Therefore, we apply the SCA method to tackle the concave term. Since the elements within the same column $d$ are coupled in $V(\mathbf{Q})$, we apply the chain rule to derive the partial derivative. By dropping the constant term $V(\mathbf{Q}^{(\iota)})$ which does not affect the optimization variables, the linear term used in the $\iota$-th iteration is formulated as
\begin{equation} \label{17}
	a\sum\limits_{d=1}^{D}{\sum\limits_{i=1}^{{{N}_{s}}}{V_{i}^{'}\left( {{\mathbf{Q}}^{(\iota )}}(i,d) \right)\left( \mathbf{Q}(i,d)-{{\mathbf{Q}}^{(\iota )}}(i,d) \right)}},
\end{equation}
where ${{\mathbf{Q}}^{(\iota )}}(i,d)$ denotes the element in the $i$-th row and $d$-th column of the matrix ${{\mathbf{Q}}^{(\iota )}}$ and for an arbitrary matrix $\mathbf{M}$ matching the dimension of $\mathbf{Q}$, the function $V_{i}^{'}\left( \mathbf{M}(i,d) \right)$ is defined as
\begin{equation}
	V_{i}^{'}\left( \mathbf{M}(i,d) \right)=\frac{{{\gamma }_{i}}}{{{\left( 1+\mathbf{M}(i,d){{\gamma }_{i}} \right)}^{3}}\sqrt{\sum\limits_{k=1}^{{{N}_{s}}}{\left( 1-\frac{1}{{{\left( 1+\mathbf{M}(k,d){{\gamma }_{k}} \right)}^{2}}} \right)}}}.
\end{equation}
}
By removing the constant part in (\ref{17}), we define
\begin{equation}
	f^{(\iota)}(\mathbf{Q})={a_{n,\varepsilon}}\sum\limits_{d=1}^{D}{\sum\limits_{i=1}^{{{N}_{s}}}{V_{i}^{'}\left( {{\mathbf{Q}}^{(\iota )}}(i,d) \right) \mathbf{Q}(i,d)}}.
\end{equation}
Then the problem (P2.1) is reformulated as the
following approximated convex optimization problem in the $\iota+1$ iteration:
\begin{subequations} \label{P3.1.1}
	\begin{align}
	{\rm P2.1.1}:\text{ }	\underset{\mathbf{Q}}{\mathop{\min }}\text { }\,&f^{(\iota)}(\mathbf{Q})-C(\mathbf{Q})+\sum\limits_{i=1}^{N_S}{\rho||\mathbf{G}(i,:)-\mathbf{Q}(i,:)||_2^2} \label{p211obj}\\ 
		\text{  s}\text{.t}\text{. }&(\ref{p31con1}), (\ref{p31con2}).
	\end{align}
\end{subequations}
Problem (P2.1.1) can be solved optimally by existing
convex optimization solvers such as CVX \cite{cvx}, and the optimal
solution yields the next iteration. The limit solution at the
convergence of the SCA also satisfies
the Karush-Kuhn-Tucker conditions of the problem
(P2.1). Thus, we can obtain the locally optimal solution of the problem (P2.1) by solving the convex problem
(P2.1.1) iteratively.

\textit{(2) Optimizing $\mathbf{G}$ with fixed $\mathbf{Q}$}: this sub-problem of (P2) is expressed as
\begin{subequations} \label{P3.2}
	\begin{align}
	({\rm P2.2}):\text{ }	\underset{\mathbf{G}}{\mathop{\min }}\text { }\,&\sum\limits_{i=1}^{N_S}{\rho||\mathbf{G}(i,:)-\mathbf{Q}(i,:)||_2^2}\\ 
		\text{  s}\text{.t}\text{. }& ||\mathbf{G}(i,:)|{{|}_{0}}\le 1,\text{ }\forall i\in \mathcal{N}_S.
	\end{align}
\end{subequations}
The optimal solution of problem (P2.2) in the $\varsigma$- iteration is \cite[prop. 3.1]{PD}
\begin{equation} \label{G}
	\mathbf{G}^{(\varsigma)}(i,d)=\delta_{i,d}\mathbf{Q}(i,d)\text{ }\forall i\in \mathcal{N}_S,\text{ }\forall d\in \mathcal{D},
\end{equation}
where 
\begin{equation}
	\delta_{i,d}=\begin{cases}
		1, & \text{if } \mathbf{Q}(i,j)=\underset{1\le j\le D}{\mathop{\max }}\,\mathbf{Q}(i,j)\\
		0, & \text{otherwise.}
	\end{cases}
\end{equation}
This means that for each $i \in \mathcal{N_S}$, the row vector $\mathbf{G}^{(\varsigma)}(i,:)$ is constructed by preserving the maximum element in $\mathbf{Q}(i,:)$ and setting all other entries in the row to 0.

\begin{algorithm}[t]
	\caption{Penalized Alternating Convex Approximation (PACA) Algorithm}
	\label{al2}
	\begin{algorithmic}[1]
		\Statex \textbf{Input:} $\bm{\gamma}$, $D$, $N_S$, $P$, $a$, $\eta$, $\epsilon$, $\epsilon_O$
		\Statex \textbf{Output:} $\mathbf{p}^{\star}$, $\mathbf{S}^{\star}$
		
		\State Initialize $\tau = 0$, $\varsigma = 0$, $\iota = 0$
		\State Initialize $\rho > 0$, $\mathbf{Q}^{(0)}$
		\State Set $f_o^{(0)} = f_m^{(0)} = f_i^{(0)} = 0$
		
		\Repeat \Comment{Outer loop: penalty update}
		\Repeat \Comment{Middle loop: alternating optimization}
		\State Update $\mathbf{G}^{(\varsigma+1)}$ according to (\ref{G})
		\Repeat \Comment{Inner loop: SCA}
		\State Solve (P2.1.1) to obtain $\mathbf{Q}^{(\iota+1)}$
		\State Compute $f_i^{(\iota+1)}$
		\State $\iota \gets \iota + 1$
		\Until $|f_i^{(\iota)} - f_i^{(\iota-1)}| \le \epsilon$
		
		\State Set $\mathbf{Q}^{(\varsigma+1)} = \mathbf{Q}^{(\iota)}$
		\State Compute $f_m^{(\varsigma+1)}$
		\State $\varsigma \gets \varsigma + 1$
		\Until $|f_m^{(\varsigma)} - f_m^{(\varsigma-1)}| \le \epsilon$
		
		\State Set $\mathbf{Q}^{(\tau+1)} = \mathbf{Q}^{(\varsigma)}$, $\mathbf{G}^{(\tau+1)} = \mathbf{G}^{(\varsigma)}$
		\State Compute $f_o^{(\tau+1)}$
		\State $\rho \gets \eta \rho$, $\tau \gets \tau + 1$
		\Until $||\mathbf{Q}^{(\tau)}-\mathbf{G}^{(\tau)}||_2^2 \le \epsilon_O$ \& $|f_o^{\left(\tau\right)}-f_o^{\left(\tau-1\right)}| \le \epsilon$
		
		\State Recover $\mathbf{p}^{\star}$ and $\mathbf{S}^{\star}$ from $\mathbf{Q}^{(\tau)}$
		\State \Return $\mathbf{p}^{\star}$, $\mathbf{S}^{\star}$
	\end{algorithmic}
\end{algorithm}

The proposed algorithm is detailed in \textbf{Algorithm~\ref{al2}}, which involves three layers of iteration. Let $\tau$, $\varsigma$, and $\iota$ denote the iteration indices of the outer, middle, and inner loops, respectively. The objective functions are denoted by $f_{o}^{(\tau)}$, $f_{m}^{(\varsigma)}$, and $f_{i}^{(\iota)}$, corresponding to problems (P1), (P2), and (P2.1.1), respectively, with their expressions given in (\ref{p2obj}), (\ref{p3obj}), and (\ref{p211obj}). In the outer iteration, the penalty factor $\rho$ is gradually increased to enforce the equivalence between problem (P1) and problem (P2). In the middle iteration, sub-problems (P2.1) and (P2.2) are solved alternately to obtain a locally optimal solution to problem (P2), where (P2.2) is convex, whereas (P2.1) is non-convex. Therefore, in the inner iteration, a locally optimal solution to problem (P2.1) is obtained by iteratively solving the approximated convex problem (P2.1.1). 

 The stopping condition of the outer iteration permits treating the converged solution $\mathbf{Q}^{(\tau)}$ as approximately feasible for problem (P1), as each row of $\mathbf{Q}^{(\tau)}$ contains at most one dominant (i.e., significantly nonzero) element, with all others being smaller than a threshold $\Gamma$, which can be arbitrarily small depending on the value of $\epsilon_O$ in the stopping criterion. Based on $\mathbf{Q}^{(\tau)}$, we recover the vector $\mathbf{p}^{\star}$ and the binary matrix $\mathbf{S}^{\star}$  as follows.  
For every $i \in \mathcal{N_S}$, 
\begin{equation}
	\mathbf{p}^{\star}(i) =
	\begin{cases}
		\mathbf{Q}^{(\tau)}(i, d), & \text{if } \exists\, d \text{ such that } \mathbf{Q}^{(\tau)}(i, d) > \Gamma \\
		0. & \text{otherwise}
	\end{cases}
\end{equation}
 Similarly, the subchannel assignment matrix $\mathbf{S}^{\star}$ is recovered as
\begin{equation}
	\mathbf{S}^{\star}(i,d) =
	\begin{cases}
		1, & \text{if } \mathbf{Q}^{(\tau)}(i,d) > \Gamma \\
		0, & \text{otherwise}
	\end{cases}
\end{equation}
for every $i \in \mathcal{N_S}$ and $d \in \mathcal{D}$.

{The convergence of the proposed PACA algorithm follows from the convergence of its three nested loops. Specifically, for a fixed $\mathbf{G}$, the inner SCA loop converges to a stationary point of problem (P2.1.1). After the inner loop converges, $\mathbf{Q}$ is updated and then alternately optimized with $\mathbf{G}$ in the middle loop until convergence. Finally, the outer loop updates the penalty parameter according to the penalty decomposition framework, whose convergence to a locally optimal solution is guaranteed by Th. 4.3 in \cite{PD}. Since each outer-level iteration is initialized with the converged solution of its inner-level iteration, the convergence of the proposed PACA algorithm is guaranteed.}

The computational complexity of interior-point methods is typically of the order $\mathcal{O}(\alpha^3)$, where $\alpha$ denotes the number of decision variables. In problem (P2.1.1), the optimization variable $\mathbf{Q} \in \mathbb{R}^{N_S \times D}$ contains a total of $\alpha = N_S D$ variables. Therefore, the worst-case per-iteration complexity of an interior-point method applied to this problem is $\mathcal{O}((N_S D)^3)$. Consequently, the overall computational complexity of the proposed algorithm is $l_1 l_2 l_3 \mathcal{O}((N_S D)^3)$, where $l_1$, $l_2$, and $l_3$ denote the numbers of iterations in the outer, middle, and inner loops, respectively. {The BMCA algorithm operates iteratively based on SCA, where a subproblem of three $N_{S}D$-dimensional variables is solved per iteration. With the subproblem solved via the interior-point method with a complexity of 
$\mathcal{O}\left((N_S D)^3\right)$, the overall complexity of BMCA until convergence is $l_{\text{B}}\mathcal{O}\left( (N_S D)^3\right)$, where $l_B$ is the number of iterations. Consequently, the PACA and BMCA algorithms are of the same order of computational complexity, as their scaling behaviors are both determined by  $\mathcal{O}\left((N_S D)^3\right)$.}

\section{STCC-MSC in Multi-user MIMO Systems}

{The achievable rate of a multi-user MIMO communication system in a single cell critically depend on how inter-user interference is managed. For traditional TCC scheme, dirty paper coding (DPC) is the only known capacity-achieving strategy for MIMO downlink (broadcast) channels, while successive interference cancellation (SIC) serves the same role in MIMO uplink (multiple-access) channels \cite{gold}. However, both DPC and SIC rely on multi-layer nonlinear processing, which makes them impractical for real-time implementation. Moreover, in the finite blocklength regime, where the communication inherently experiences a nonzero error probability, the error propagation effect in SIC and DPC will be amplified, making it non-negligible. As a low-complexity and widely used alternative, block diagonalization (BD) linearly eliminates inter-user interference by projecting each user’s signal onto the null space of others, effectively transforming the multi-user MIMO channel into a set of parallel P2P MIMO channels \cite{ZF1}. Although BD is sub-optimal in terms of capacity, it can asymptotically approaches the capacity in the high-SNR regime \cite{ZF1}, and its low implementation complexity makes it attractive for practical systems.}

{
Other transceiver designs such as successive zero-forcing DPC \cite{ZFSIC}, weighted MMSE \cite{WMMSE}, and related approaches \cite{IC} achieve better interference suppression–noise amplification trade-offs in certain regimes, but typically incur higher computational complexity and still leave residual interference. To date, information-theoretic results for STCC are only available for parallel Gaussian channels. Therefore, this section adopts BD to eliminate inter-user interference, thereby enabling the application of the STCC-MSC architecture and the joint subchannel and power allocation strategy proposed in Sections \uppercase\expandafter{\romannumeral3} and \uppercase\expandafter{\romannumeral4}.}

{Let $K$ denote the number of users.} The $k$-th user is equipped with $N_k$ antennas. The set of users is denoted by $\mathcal{K}=\left\{1,2,...,K\right\}$. The base station (BS) is equipped with $N_B$ antennas, satisfying $\sum\limits_{k=1}^K{N_k} \le N_B$.  In the following two sub-sections, we separately consider the downlink and uplink cases. For clarity, the same notation \( \mathbf{H}_k \) is used to represent the channel matrix between the BS and the \( k \)-th user. Specifically, \( \mathbf{H}_k \in \mathbb{C}^{N_k \times N_B} \) in the downlink and \( \mathbf{H}_k \in \mathbb{C}^{N_B \times N_k} \) in the uplink. The allowable number of streams of user $k$ is $D_k$, satisfying $D_k \le N_{S,k} \le N_k$, where $N_{S,k}={\rm rank}(\mathbf{H}_k)$ is the $k$-th user's spatial DoF. The set of subchannels and streams for user $k$ are denoted by $\mathcal{N}_{S,k}=\left\{1,2,...,N_{S,k}\right\}$ and $\mathcal{D}_k=\left\{1,2,...,D_k\right\}$, respectively.

\subsection{Multi-User Downlink STCC-MSC}

Let $\mathbf{A}_k\overline{\mathbf{x}}_k \in \mathbb{C}^{N_B \times 1}$ denote the transmitted symbol vector for user $k$ at a channel use, where $\mathbf{A}_k$ is the inter-user interference elimination matrix for user $k$. The transmitted symbol vector at the BS is $\sum\limits_{k=1}^K{\mathbf{A}_k\overline{\mathbf{x}}_k}$. The received signal at the $k$-th user is 
\begin{equation} \label{dl}
	\overline{\mathbf{y}}_k=\mathbf{H}_k\mathbf{A}_k\overline{\mathbf{x}}_k+\sum\limits_{j=1,j\ne k}^{K}{\mathbf{H}_k\mathbf{A}_j\overline{\mathbf{x}}_j}+\mathbf{w}_k,
\end{equation}
where $\mathbf{w}_k \sim \mathcal{CN}(0, \sigma_k^2\mathbf{I}_{N_k \times N_k})$ is the AWGN. In the downlink communication, the BD method projects the signal intended for user $k$ onto the null space of the channels of all other users. Defining
\begin{equation} \label{BD1}
	\overline{\mathbf{H}}_k=\left[\mathbf{H}_1^T,...,\mathbf{H}_{k-1}^T,\mathbf{H}_{k+1}^T,...,\mathbf{H}_K^T\right]^T \in \mathbb{C}^{\overline{N}_k \times N_B},
\end{equation}
which denotes the combined channel matrix of all users except user $k$,  where $\overline{N}_k=\sum\limits_{j=1,j \ne k}^K{N_j}$, satisfying $N_B-\overline{N}_k \ge N_k$ because $\sum\limits_{k=1}^K{N_k} \le N_B$. The SVD of $\overline{\mathbf{H}}_k$ is expressed as
\begin{equation}\label{BD2}
	\overline{\mathbf{H}}_k=\overline{\mathbf{U}}_k\left[ \begin{matrix}
		{{\overline{\mathbf{\Lambda }}}_{k}} & \mathbf{0}  \\
		\mathbf{0} & \mathbf{0}  \\
	\end{matrix} \right]
	{{\left[ \begin{matrix}
				\overline{\mathbf{V}}_{k}^{(1)} & \overline{\mathbf{V}}_{k}^{(0)}  \\
			\end{matrix} \right]}^{H}},
\end{equation}
where $\overline{\mathbf{\Lambda}}_k \in \mathbb{R}^{\overline{N}_k \times \overline{N}_k}$ is the diagonal matrix whose diagonal entries are the singular values. $\overline{\mathbf{V}}_{k}^{(1)}$ and $\overline{\mathbf{V}}_{k}^{(0)}$ are  matrices with $N_B$ rows, containing the first $\overline{N}_k$  and the last $\left(N_B-\overline{N}_k\right)$ right singular vectors, respectively. Since the columns of $\overline{\mathbf{V}}_{k}^{(0)}$ form an orthonormal basis for the null space of $\overline{\mathbf{H}}_k$, we select its first $N_k$ columns to construct the inter-user interference elimination matrix for user $k$, denoted by $\mathbf{A}_k \in \mathbb{C}^{N_B \times N_k}$. The matrices $\left\{ \mathbf{A}_k \right\}_{k=1}^K$ satisfy that $\mathbf{H}_k\mathbf{A}_j=\mathbf{0}, \text{ } \forall j \ne k$. Therefore, the channel (\ref{dl}) can be rewritten as
\begin{equation} \label{dl1}
	\overline{\mathbf{y}}_k=\mathbf{H}_k\mathbf{A}_k\overline{\mathbf{x}}_k+\mathbf{w}_k,
\end{equation}
where $\overline{\mathbf{x}}_k \in \mathbb{C}^{N_k \times 1}$. 
The multi-user MIMO downlink system is now decoupled into \( K \) parallel point-to-point sub-systems, where each sub-system can make full use of its spatial DoF, because \( N_k \ge N_{S,k} \).

Since $\mathbf{A}_k$ consists of a set of orthonormal basis, multiplying a signal vector by $\mathbf{A}_k$ does not change the transmit power, i.e., $||\mathbf{A}_k\overline{\mathbf{x}}_k||_2^2=||\overline{\mathbf{x}}_k||_2^2$. Therefore, for each user, we obtain a P2P MIMO communication sub-system by equivalently treating \( \mathbf{H}_k \mathbf{A}_k \) in (\ref{dl1}) as the effective channel matrix in (\ref{n1}). Based on this, the designs presented in Sections~\uppercase\expandafter{\romannumeral3} and~\uppercase\expandafter{\romannumeral4} can be applied to obtain the achievable rate of the downlink multi-user STCC-MSC. Each sub-system follows the transmission architecture illustrated in Fig.~\ref{system}, where the left singular matrix \( \mathbf{U}_k \in \mathbb{C}^{N_k \times N_{S,k}} \) and the right singular matrix \( \mathbf{V}_k \in \mathbb{C}^{N_B \times N_{S,k}} \) consist of the first \( N_{S,k} \) columns of the singular matrices obtained from the SVD of \( \mathbf{H}_k \mathbf{A}_k \) for the \( k \)-th sub-system. $\bm{\lambda}_k=\left[\lambda_{k,1}, \lambda_{k,2}...,\lambda_{k,N_{S,k}}\right]^T$ are the ordered non-zero singular values of the matrix $\mathbf{H}_k \mathbf{A}_k\left(\mathbf{H}_k \mathbf{A}_k\right)^H$. 

Let $\mathbf{p}_k=\left[p_{k,1},p_{k,2},...,p_{k,N_{S,k}}\right]^T$ denote the power allocation vector for the $k$-th sub-system and $\gamma_{k,i}=p_{k,i}/\sigma_k^2,\text{ }\forall k \in \mathcal{K}, \forall i\in \mathcal{N}_{S,k}$. The maximum transmit power at the BS is $P_D$. Recalling problem (\ref{P1}), the subchannel assignment matrices $\left\{\mathbf{S}_k \in \mathbb{R}^{N_{S,k}\times D_k}\right\}_{k=1}^K$ and the power allocation vectors $\left\{\mathbf{p}_k \in \mathbb{R}^{N_{S,k}\times 1}\right\}_{k=1}^K$ are determined by the solution to the following sum-rate maximization problem:
\begin{subequations} \label{PDL}
	\begin{align}
		\underset{\left\{\mathbf{p}_k\right\}_{k=1}^K,\left\{\mathbf{S}_k\right\}_{k=1}^K}{\mathop{\min }}\,& -\sum\limits_{k=1}^{K}{\overline{R}_k} \label{dlobj} \\ 
		\hspace{-0.9em}\text{  s}\text{.t}\text{. }&\sum\limits_{k=1}^{K}\sum\limits_{i=1}^{{{N}_{S,k}}}{{{p}_{k,i}}\le P_D},\text{    } \label{dlcon1}\\ 
		&\hspace{-0.9em} {{p}_{k,i}}\ge 0,\text{ }\forall k \in \mathcal{K}, \forall i\in \mathcal{N}_{S,k}, \label{dlcon2}\\ 
		&\hspace{-0.9em} \mathbf{S}_k(i,d)=\{0,1\},\forall k \in \mathcal{K},\forall i\in \mathcal{N}_{S,k}, \forall d\in \mathcal{D}_k, \label{dlcon3}\\ 
		&\hspace{-0.9em} \sum\limits_{d=1}^{D_k}\mathbf{S}_k(i,d) \le 1, \text{ }\forall k \in \mathcal{K}, \text{ }\forall i\in \{1,...,{{N}_{S,k}}\}, \label{dlcon4}
	\end{align}
\end{subequations}
where
\begin{align}
	-\overline{R}_k &= \sum\limits_{d=1}^{D_k} \left( {a_{n,\varepsilon}}\sqrt{ \sum\limits_{i=1}^{N_{S,k}} \mathbf{S}_k(i,d)V_{k,i}(p_{k,i}) } \right. \nonumber \\
	&\quad \left. - \sum\limits_{i=1}^{N_{S,k}} \mathbf{S}_k(i,d)\log \left( 1 + p_{k,i} \gamma_{k,i} \right) \right)
\end{align}
and
\begin{equation}
	{{{V}}_{k,i}}=1-\frac{1}{{{\left( 1+{{{p}_{k,i}}{{\gamma}_{k,i}}} \right)^2}}}. 
\end{equation}

The problem (\ref{PDL}) can be viewed as a multi-user generalization of problem (\ref{P1})—essentially a $K$-fold higher-dimensional extension, where the total transmit power is globally constrained by (\ref{dlcon1}). Therefore, the proposed PACA algorithm in Section \uppercase\expandafter{\romannumeral4}, developed for solving (\ref{P1}) can applied to solve (\ref{PDL}). Specifically, we can first introduce the auxiliary variables $\left\{\mathbf{Q}_k\right\}_{k=1}^{K}$, where the element in the $i$-th row and $d$-th column $\mathbf{Q}_k(i,d)=\mathbf{S}_k(i,d)p_{k,i}$. Then by further  introducing the auxiliary variables  $\left\{\mathbf{G}_k\right\}_{k=1}^{K}$ and corresponding penalty factors $\left\{\rho_k\right\}_{k=1}^{K}$, the problem (\ref{PDL}) can be reformulated as a continuous non-convex problem similar in structure to problem (P2) in (\ref{P3}), with variables to be $\left\{\mathbf{Q}_k\right\}_{k=1}^{K}$ and $\left\{\mathbf{G}_k\right\}_{k=1}^{K}$. Finally, this reformulated problem can be effectively solved by applying the proposed PACA algorithm, as detailed in \textbf{Algorithm~\ref{al2}}, after replacing the inputs and variables associated with P2P STCC-MSC by those corresponding to the multi-user setting. As a result, the overall computational complexity for solving (\ref{PDL}) scales linearly with the number of users $K$.

\subsection{Multi-User Uplink STCC-MSC}
Due to the uplink-downlink duality, similarly, the BD method is used to manage the inter-user interference in the uplink communication system. With a slight abuse of notations, we use $\overline{\mathbf{x}}_k \in \mathbb{C}^{N_k \times 1}$ to denote the transmitted symbol vector at user $k$ at a channel use. The received signal at the BS is
\begin{equation}
	{\mathbf{y}}=\sum\limits_{k=1}^{K}{\mathbf{H}_k\overline{\mathbf{x}}_k}+{\mathbf{w}},
\end{equation}
where ${\mathbf{w}} \sim \mathcal{CN}(0, \sigma^2\mathbf{I}_{N_B \times N_B})$ is the AWGN at the BS. By left-multiplying $\mathbf{B}_k \in \mathbb{C}^{N_k \times N_B}$, we can recover the signal of the $k$-th user as
\begin{equation} \label{34}
	\mathbf{B}_k{\mathbf{y}}=\mathbf{B}_k\mathbf{H}_k\overline{\mathbf{x}}_k+\mathbf{B}_k\sum\limits_{j=1,j\ne k}^{K}{\mathbf{H}_j\overline{\mathbf{x}}_j}+\mathbf{B}_k{\mathbf{w}}.
\end{equation}
Following the steps in (\ref{BD1}) and (\ref{BD2}), we apply the SVD to \( \overline{\mathbf{H}}_k^H \), where \( \overline{\mathbf{H}}_k \) is the concatenated channel matrix of all users except user \( k \). This yields an orthonormal basis for the null space of \( \overline{\mathbf{H}}_k^H \). The unitary matrix \( \mathbf{B}_k \) is then constructed by selecting the first \( N_k \) columns from this basis and applying the Hermitian transpose. As a result, we have \( \mathbf{B}_k \overline{\mathbf{H}}_k = \overline{\mathbf{H}}_k^H \mathbf{B}_k^H = \mathbf{0}_{N_k \times N_k} \), and (\ref{34}) is rewritten as
\begin{equation} \label{35}
	\overline{\mathbf{y}}_k=\mathbf{B}_k\mathbf{H}_k\overline{\mathbf{x}}_k+\overline{\mathbf{w}},
\end{equation}
where $\overline{\mathbf{y}}_k=\mathbf{B}_k{\mathbf{y}}$ and $\overline{\mathbf{w}}=\mathbf{B}_k{\mathbf{w}}$. Since $\mathbf{B}_k$ is unitary, we have $||\overline{\mathbf{y}}_k||^2_2=||{\mathbf{y}}_k||^2_2$ and ${\overline{\mathbf{w}}} \sim \mathcal{CN}(0, \sigma^2\mathbf{I}_{N_k \times N_k})$.

The multi-user MIMO uplink system is now decoupled into \( K \) parallel P2P sub-systems with $\mathbf{B}_k\mathbf{H}_k$ to be the equivalent channel matrix in (\ref{n1}). Similar as in the downlink communication, each sub-system follows the architecture illustrated in Fig. \ref{system}, where the singular matrices are obtained from the SVD of $\mathbf{B}_k\mathbf{H}_k$ for the $k$-th sub-system. In contrast to the downlink case where joint power allocation is required at the BS after applying BD, the uplink setting allows each user to independently maximize its own transmission rate due to the absence of inter-user interference.  Therefore, solving the sum-rate maximization problem in multi-user uplink STCC-MSC is equivalent to solving \( K \) parallel rate maximization problems. By reusing the notations in Section \uppercase\expandafter{\romannumeral5}-A, the \( k \)-th problem given by
\begin{subequations} \label{PUL}
	\begin{align}
		\underset{\mathbf{p}_k,\mathbf{S}_k}{\mathop{\min }}\,&-\overline{R}_k\label{PULobj} \\ 
		\text{  s}\text{.t}\text{. }&\sum\limits_{i=1}^{{{N}_{S,k}}}{{{p}_{k,i}}\le P_k},\text{    }\\ 
		& {{p}_{k,i}}\ge 0,\text{ }\forall i\in \mathcal{N}_{S,k},\\ 
		& \mathbf{S}_k(i,d)=\{0,1\},\text{ }\forall i\in \mathcal{N}_{S,k},\text{ } \forall d\in \mathcal{D}_k, \\ 
		& \sum\limits_{d=1}^{D_k}\mathbf{S}_k(i,d) \le 1, \text{ }\forall i\in \mathcal{N}_{S,k}, 
	\end{align}
\end{subequations}
where $P_k$ is the maximum transmit power of user $k$ and $p_{k,i}$ is the power allocated to its $i$-th subchannel by user $k$. Problem (\ref{PUL}) can be regarded as a P2P STCC-MSC rate maximization problem and can be effectively solved by directly applying the proposed PACA algorithm presented in Section \uppercase\expandafter{\romannumeral4}.

\section{Numerical Results}
\newcommand{\tabincell}[2]{\begin{tabular}{@{}#1@{}}#2\end{tabular}}
\begin{table}[t]
	\centering
	\renewcommand\arraystretch{1.3}
	\caption{Simulation Parameters}
	\label{tab1}
	\setlength{\tabcolsep}{3pt}
	\begin{tabular}{|p{120pt}<{\centering}|p{100pt}<{\centering}|}
		\hline
		\textbf{Parameters} & \textbf{Value Range}\\
		\hline
		Carrier frequency ${f_c}$	&  3.5GHz\\
		\hline
		Transmission Bandwidth ${B}$	&  60 MHz \cite{3GPP2}\\
		\hline
		Noise power spectral density & -174 dBm/Hz \\
		\hline
		Number of channel uses \textit{n} &30\\
		\hline
		Block error probability $\varepsilon$ & $10^{-6}$\\
		\hline
		Initial penalty factor $\rho$ & $1$\\
		\hline
		Penalty update factor $\eta$ & $2$\\
		\hline
	\end{tabular}
\end{table}
This section presents numerical results for both P2P (Figs. \ref{convergence}–\ref{imp}) and multi-user (Fig. \ref{mu}) scenarios \footnote{The simulation code will be provided to reproduce the results of this paper once accepted: https://github.com/L-Jin-bupt}. Following prior works \cite{CL,FY,TVT}, wireless channels are modeled as Rayleigh fading, where each element of the channel matrix is independently drawn from a circularly symmetric complex Gaussian distribution $\mathcal{CN}(0, \beta)$, with $\beta$ denoting the large-scale path loss. The path loss $\beta$ is set as $32.4+20{\rm log}10(f_c)+30{\rm log}10(\overline{d})$  \cite{3GPP1}, where $f_c$ is the carrier frequency and $\overline{d}$ is the distance between the transmitter and the receiver. {All simulation results are averaged over 5000 independent channel realizations.} {Since the exact finite-$n$ remainder term is unavailable in closed-form, and the commonly used approximation $\frac{{\rm log}n}{2n}$ is relatively small and
	would only introduce a constant offset for a fixed $n$, all numerical results are obtained using the standard normal approximation introduced in  Section \uppercase\expandafter{\romannumeral2} and Section \uppercase\expandafter{\romannumeral3}, to maintain consistency with the optimization framework. This omission does not affect the relative performance comparisons or the conclusions drawn from the numerical results.}

In the P2P communication, the antenna numbers of transmitter and receiver are set as 12 and 8, respectively, and the distance between them is $\overline{d}=200$ meters. The maximum transmit power is set to 24 dBm, satisfying the allowable transmit power for local area BSs \cite{3GPP2}. In the multi-user scenario, 8 users are considered. The BS is equipped with 32 antennas, and each user has 4 antennas. The allowable number of streams for each user is set to 2. The users are uniformly distributed within a circular area of radius 50 meters, with the circle's center located 200 meters from the BS. For downlink communication, the BS's maximum transmit power is set to 28 dBm, satisfying the allowable transmit power for medium range BSs \cite{3GPP2}. For uplink communication, the total maximum transmit power across all users is 28 dBm, equally allocated among them. Other simulation parameters are shown in Table \ref{tab1}, if not specifically mentioned.

{To enable a fair comparison with the proposed STCC architecture under the same number of transmitted streams $D$, we introduce a benchmark termed limited-stream temporal channel coding (LS-TCC), where the conventional TCC is constrained to employ at most $D$ streams. Under this constraint, as we described in Appendix \ref{appa2}, only the $D$ strongest eigenchannels are activated, and the transmit power is allocated among them. In addition, the conventional TCC without this stream constraint is also included as a benchmark, where all available spatial subchannels are utilized i.e., $D=N_S$.}

The main labels in the figures are explained as: 
\begin{enumerate}
	\item “STCC-PACA" represents the STCC-MSC scheme using the PACA algorithm proposed in this paper. 
	\item “STCC-BMCA" represents the STCC-MSC scheme using the BMCA algorithm proposed in \cite{PE}. 
	\item { “TCC-WF" represents the TCC scheme employing water-filling (WF) method for power allocation, which achieves capacity in the asymptotic regime but is not valid in the finite blocklength regime.}
	\item {“TCC-SCA" represents the TCC scheme where the power allocation uses the successive convex approximation (SCA) method, which is described in Appendix \ref{appa2}.}
	\item {“LS-TCC-WF" represents the LS-TCC benchmark employing the WF power allocation.}
	\item {“LS-TCC-SCA" represents the LS-TCC benchmark employing the SCA power allocation.}
\end{enumerate}

\begin{figure}[t!]
	\centerline{\includegraphics[width=1\linewidth]{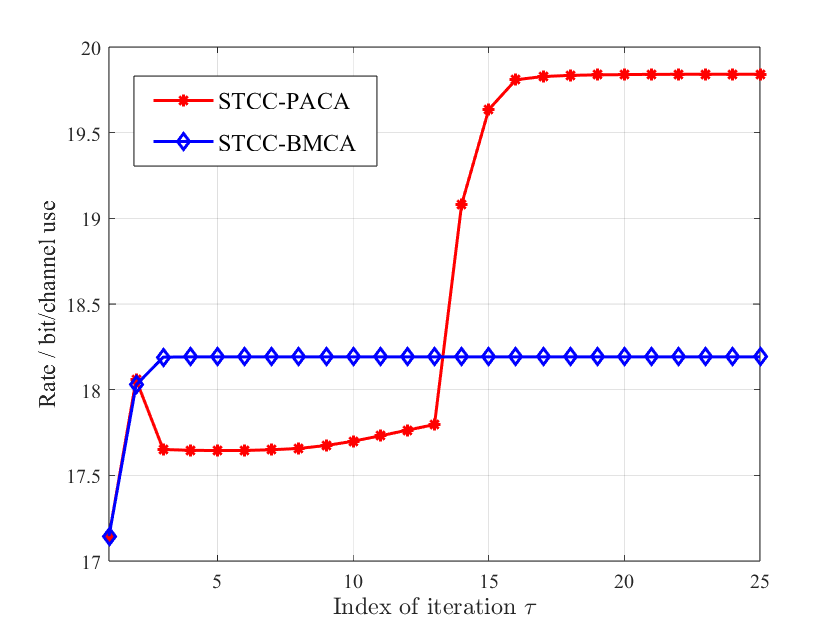}}
	\caption{Rate versus index of iteration $\tau$.}
	\label{convergence}
\end{figure}

Fig. \ref{convergence} demonstrates the convergence of the proposed PACA algorithm and the benchmark BMCA algorithm in solving the STCC rate maximization problem (\ref{P1}) with the allowable number of streams to be 5. As shown in the figure, the PACA algorithm exhibits slower convergence, but eventually attains a higher achievable rate compared to the BMCA algorithm. In the PACA algorithm, the objective function exhibits fluctuations in early iterations due to a small penalty factor $\rho$, which leads to intermediate solutions violating the $\ell_0$-norm constraint (\ref{p2con2}) in problem (P2). As $\rho$ increases exponentially, the objective becomes monotonic and grows rapidly after crossing a threshold, eventually converging.

Fig. \ref{versusD} plots the rate versus the allowable number of streams, $D$. The results indicate that the PACA algorithm achieves a 9.85\% rate improvement over the BMCA algorithm within the STCC scheme. Furthermore, the STCC-PACA combination improves the rate by 28.68\% over the LS-TCC scheme. 	{ Since TCC is a special case of LS-TCC when $D = N_S$, its achievable rate remains invariant with respect to $D$. Unlike the monotonic decreasing trend observed in \cite{TVT} under CSI-R and linear receivers, the achievable rate obtained by the proposed PACA algorithm does not strictly decrease with $D$. This difference arises because, under the full-CSI setup, the transmit power and stream assignment are jointly optimized, resulting in different resource allocation solutions for different $D$. Moreover, since problem (\ref{P1}) is non-convex, PACA converges to a locally optimal solution, whose objective value may exhibit non-monotonic variations with respect to $D$. In contrast, the BMCA algorithm exhibits a monotonically decreasing rate as $D$ increases, indicating a different optimization behavior.} 

\begin{figure}[t!]
	\centerline{\includegraphics[width=1\linewidth]{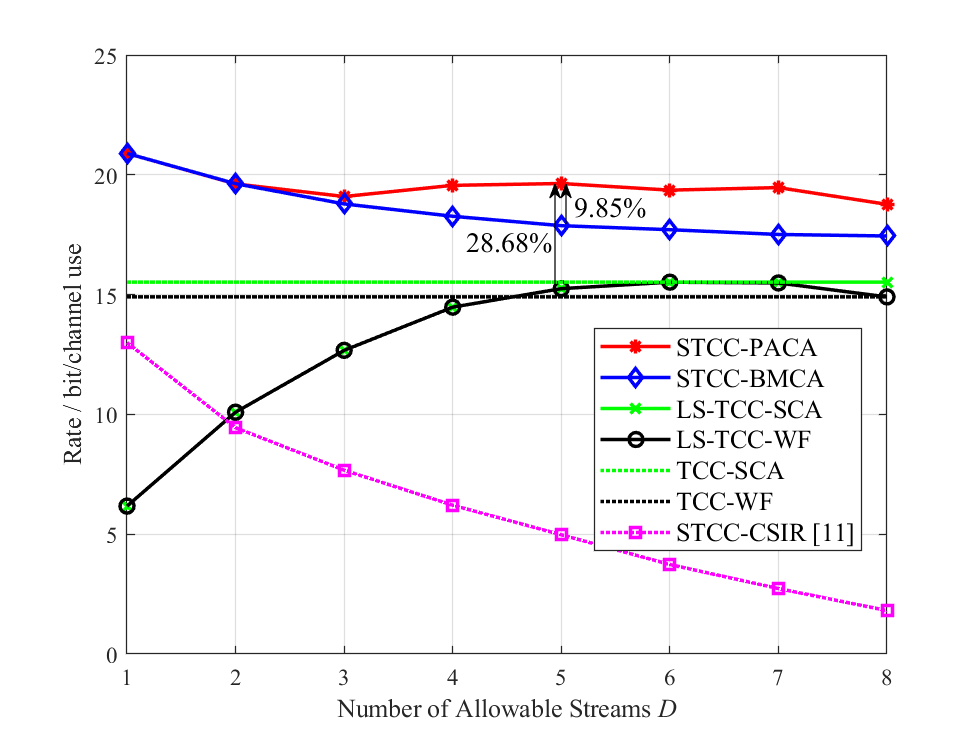}}
	\caption{{Rate versus allowable number of streams $D$.}}
	\label{versusD}
\end{figure}

\begin{figure}[t!]
	\centerline{\includegraphics[width=1\linewidth]{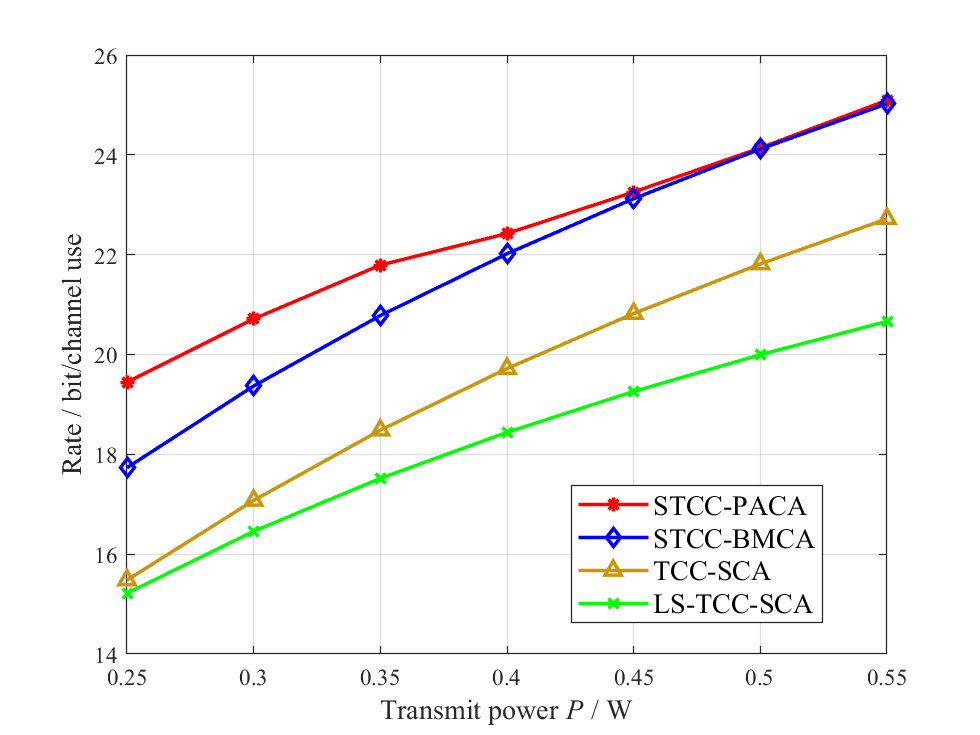}}
	\caption{{Rate versus transmit power $P$.}}
	\label{versusP}
\end{figure}

\begin{figure}[t!]
	\centerline{\includegraphics[width=1\linewidth]{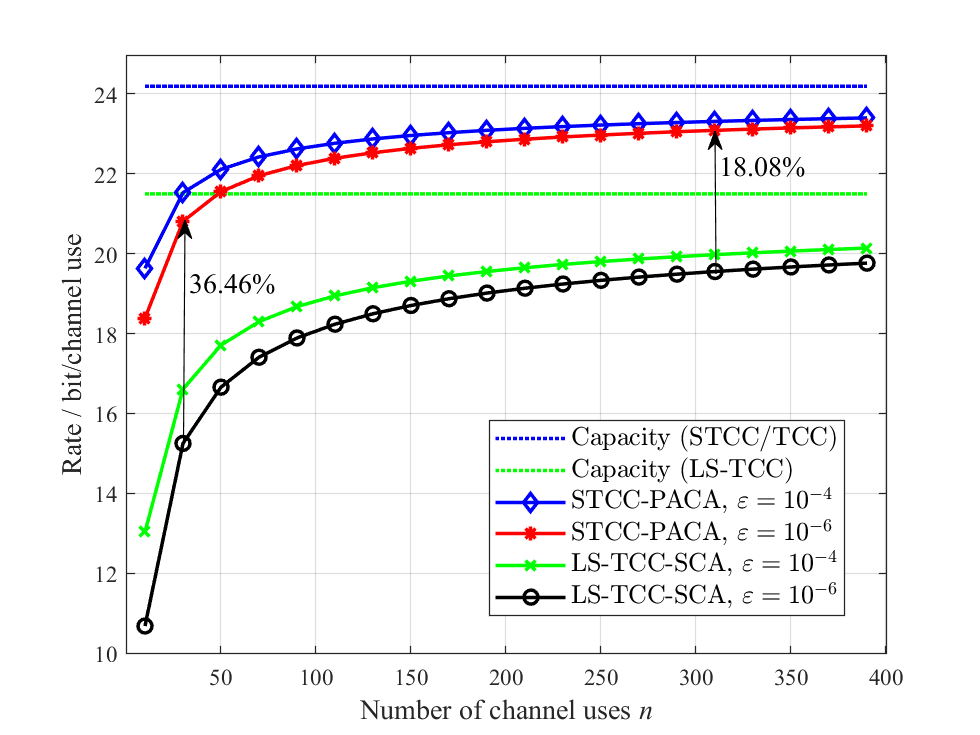}}
	\caption{{Rate versus number of channel uses $n$.}}
	\label{changen}
\end{figure}

Fig.~\ref{versusP} plots the achievable rate versus the transmit power $P$, with the allowable number of streams $D = 5$. {As WF and SCA methods yield almost equal rates for both TCC and LS-TCC when $D=5$, only TCC-SCA and LS-TCC-SCA are plotted in Fig. \ref{versusP}.} As expected, all schemes achieve higher rates as the transmit power increases. An interesting observation is that the performance advantage of the PACA algorithm over the BMCA algorithm gradually diminishes with increasing transmit power. This is mainly because the performance gain of PACA stems from its ability to reduce the rate loss term involving $V_i$ in the objective function~(\ref{p2obj}). As the SNR increases, the channel dispersion term $V_i$, defined in~(\ref{Vi}), approaches 1, rendering the rate loss term a constant. As a result, the rates achieved by the two algorithms gradually converge. {Moreover, increasing the transmit power amplifies the rate advantage of TCC over LS-TCC, which fundamentally originates from TCC's full subchannel multiplexing gain.}

Fig. \ref{changen} plots the rate versus the number of channel uses $n$ under different error probability requirements, with $D = 5$. { As shown in Fig. \ref{versusD}, the achievable rates of LS-TCC and TCC are almost identical under this setting. Hence, only the LS-TCC curves are shown for visual clarity. As indicated by equations (\ref{RT}) and (\ref{Rd}), in both STCC and TCC schemes, the rate increases with larger $n$ and decreases as the error probability requirement becomes more stringent.  Further, we investigated the asymptotic capacity when $n \to \infty$ as described in App. \ref{appa3}, where we prove that STCC has the same asymptotic capacity as TCC. It can be observed from the figure that as $n$ increases, the achievable rates of all schemes gradually approach their corresponding asymptotic capacities. Moreover, consistent with the findings in \cite{FY,TWC,TVT}, it can be observed that the smaller the $n$, the greater the rate advantage of STCC over LS-TCC, highlighting its superiority in short-blocklength communications.}

\begin{figure}[t!]
	\centerline{\includegraphics[width=1\linewidth]{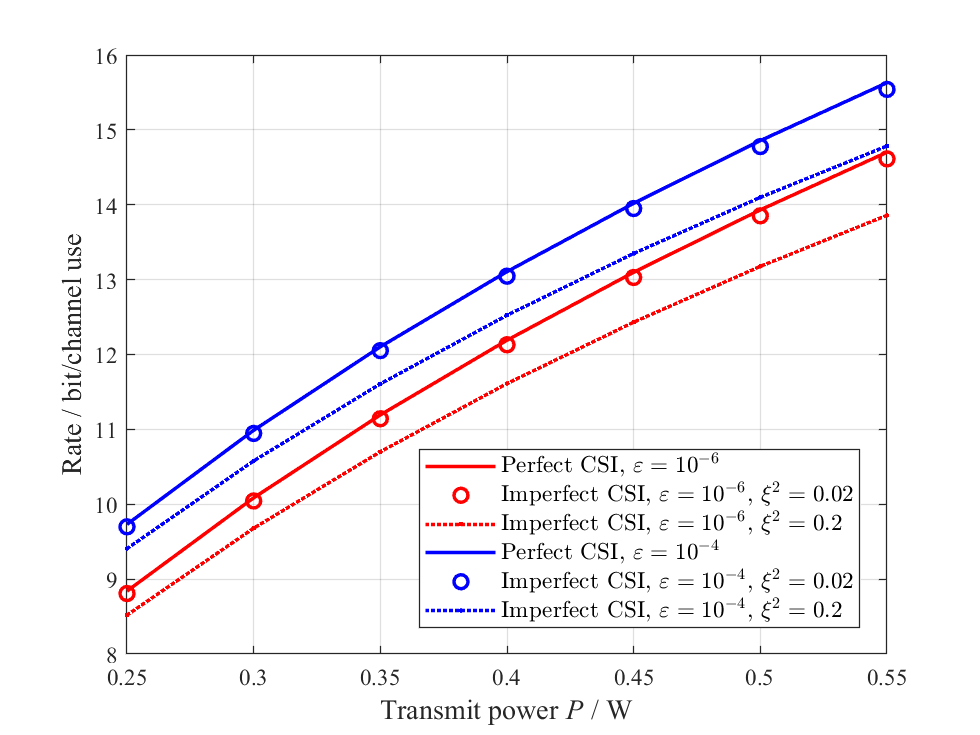}}
	\caption{{Rate versus transmit power under imperfect CSI.}}
	\label{imp}
\end{figure}

{Fig. \ref{imp} illustrates the impact of imperfect CSI with $D=5$.  Under imperfect CSI, the impact introduced by the channel estimation error is equivalently treated as an additional noise, following the methodology in \cite{CL}. 
Specifically, the thermal noise power is normalized to unity, while the estimation-error variance $\xi^2$ is set to 0.02 and 0.2. Under this model, the effective received SNR of the $i$-th subchannel is given by $\frac{p_i/\sigma^2}{p_i\xi^2+1}$.
As expected, channel estimation errors reduce the effective received SNR, resulting in a lower achievable rate. Moreover, this performance degradation becomes more pronounced at higher transmit SNRs, which is consistent with the observations reported in \cite{CL}.}

\begin{figure}[t!]
	\centerline{\includegraphics[width=1\linewidth]{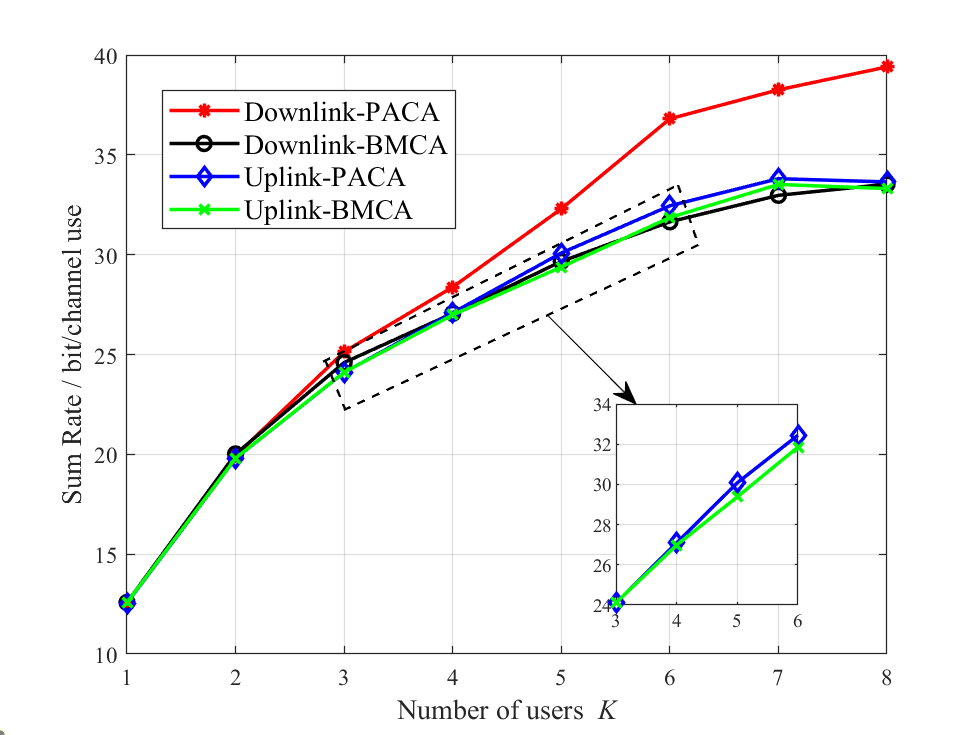}}
	\caption{Sum-rate versus number of users $K$.}
	\label{mu}
\end{figure}

Fig. \ref{mu} plots the sum rate of all users in both downlink and uplink communications. The downlink achieves a higher sum rate than the uplink when serving the same number of users, and this advantage becomes more significant with increasing user numbers. This is attributed to the centralized subchannel and power allocation at the BS in the downlink communication, whereas uplink users optimize their own transmissions independently, lacking coordination. As the number of users increases, the sum rate generally improves due to better utilization of spatial DoF, which leads to higher spatial multiplexing gains. However, when the total number of user antennas approaches that of the BS, the null space available to each user under the BD method becomes significantly restricted. This limitation reduces the effective transmission capability, leading to a direct degradation in the uplink sum rate. In contrast, the downlink benefits from centralized resource allocation at the BS, which, despite yielding diminishing rate gains, ensures that the sum rate remains increasing. Moreover, as the number of users increases, the transmit power allocated to each user decreases, thereby enlarging the rate advantage of the PACA algorithm over the BMCA algorithm, which is consistent with the findings shown in Fig. \ref{versusP}.

\section{Conclusions and Discussions}
	This paper investigate the achievable rate of STCC in MIMO systems when CSI is available in the finite blocklength regime. We first build the STCC-MSC transmission architecture in the P2P MIMO system. Then we maximize the rate of STCC-MSC by jointly optimizing the subchannel assignment and power allocation strategies, which is formulated as a MINLP problem. A PACA algorithm is proposed to efficiently solve the problem and obtain a locally optimal solution. We further extend the designs to the more general multi-user MIMO systems. By applying BD method to eliminate inter-user interference, the reformulated sum-rate maximization problem can be effectively solved by the proposed PACA algorithm. Simulation results indicate that the PACA algorithm achieves a 9.85\% rate improvement over the benchmark algorithm within the STCC-MSC scheme. Furthermore, the joint STCC-MSC-PACA scheme improves the rate by 28.68\% over traditional TCC.

{This work provides an application of STCC in MIMO communication systems. Compared with the traditional TCC architecture, the proposed STCC-MSC architecture incorporates a stream–subchannel matching mechanism with modest modifications to the overall structure.
Moreover, the PACA algorithm provides a suitable basis for exploring further extensions in future work. For example, these extensions may involve different optimization objectives, such as maximizing the weighted sum rate or the minimum user rate to enhance fairness; additional constraints, such as minimum rate guarantees for specific users or streams; and the application of STCC in emerging communication paradigms, including cell-free systems and reconfigurable intelligent surface-assisted communications.}

\appendices
\section{{Analysis of the Spurious Optimization of $\mathbf{S}$ with Fixed $\mathbf{q}$ in Problem (\ref{P1})}}
\label{appa1}
When using the alternating optimization to optimize $\mathbf{S}$ and $\mathbf{p}$ iteratively to solve problem (\ref{P1}), the sub-problem of solving $\mathbf{S}$ with fixed $\mathbf{q}$ is expressed as
\begin{subequations} \label{p1}
	\begin{align}
		\underset{\mathbf{S}}{\mathop{\min }}\,&f_V(\mathbf{S})-f_C(\mathbf{S}) \label{objo} \\ 
		\text{  s}\text{.t}\text{. }& \mathbf{S}(i,d)=\{0,1\},\text{ }\forall i\in \mathcal{N}_S,\text{ } \forall d\in \mathcal{D}, \label{co3}\\ 
		& \sum\limits_{d=1}^{D}\mathbf{S}(i,d) \le 1, \text{ }\forall i\in \mathcal{N}_S, \label{co4}.
	\end{align}
\end{subequations}
where 
\begin{equation}
	f_V(\mathbf{S})={a_{n,\varepsilon}}\sum\limits_{d=1}^{D}{\sqrt{\sum\limits_{i=1}^{{{N}_{S}}}{\mathbf{S}(i,d)V_i(p_i)}}} 
\end{equation}
and
\begin{equation}
	f_C(\mathbf{S}) = \sum\limits_{d=1}^{D}{\sum\limits_{i=1}^{{{N}_{S}}}{\mathbf{S}(i,d)\log \left( 1+{{p}_{i}}{{\gamma }_{i}} \right)} }.
\end{equation}

The constraints (\ref{co3}) and (\ref{co4}) imply that the optimal solution $\mathbf{S}^*$ to problem (\ref{p1}) should approximately satisfy
\begin{equation}
	f_C\left(\mathbf{S}^*\right) \approx \sum\limits_{i=1}^{{N}_{S}} \log \left( 1 + p_i \gamma_i \right), \label{ap}
\end{equation}
which is a constant and suggests that all subchannels are used for transmission. In general, the more subchannels are used, the smaller the objective function (\ref{objo}) becomes. However, there may exist a small number of subchannels with extremely low signal-to-noise ratios (SNRs), i.e., \( p_i \gamma_i \), for which including them in the transmission process may actually increase the value of (\ref{objo}) because the rates of these subchannels may be less than 0 \cite[Fig. 3]{mono}. Nevertheless, since their SNRs are very low, their contribution to \( f_C \) is negligible, and thus the approximation in (\ref{ap}) remains valid.

Therefore, the objective function of problem (\ref{p1}) becomes minimizing $f_V(\mathbf{S})$. Due to the Jensen's inequality, the optimal solution $\mathbf{S}^*$  should satisfy that 
\begin{equation}
	\left| \left\{ j \in \{1, 2, \dots, D\} \; \middle| \; \mathbf{S}^*(:,j) \neq \mathbf{0} \right\} \right| = 1, \label{sas}
\end{equation}
which means that $\mathbf{S}^*$ has exactly one nonzero column, and all other columns are zero vectors. The physical meaning of (\ref{sas}) is that the optimal solution to problem (\ref{p1}) requires all subchannels to serve one stream. Therefore, if we directly perform alternating optimization between $\mathbf{p}$ and $\mathbf{S}$, the sub-problem of optimizing the subchannel assignment matrix $\mathbf{S}$ yields a solution that violates the intended meaning of the original problem. Hence, we refer to this sub-problem as spurious.

\section{SCA Method for Maximizing the Rate of TCC}
\label{appa2}
In the TCC scheme, to maximize the MIMO system rate, when the allowable number of streams $D$ is less than the spatial DoF $N_S$, the $D$ subchannels with the highest SNRs are used for transmission, sharing the total transmit power $P$. Let $\overline{\mathbf{p}}=[\overline{p}_1,\overline{p}_2,...,\overline{p}_D]^T$ denote the allocated power to subchannels/streams. Recalling (\ref{RT}), the problem of optimizing the power allocation to maximize the rate of TCC is formulated as
\begin{subequations} \label{PT}
	\begin{align}
		\underset{\overline{\mathbf{p}}}{\mathop{\min}} \quad 
		& \sum_{i=1}^{D} \left( {a_{n,\varepsilon}}\sqrt{V_i(\overline{p}_i)} - \log(1 + \overline{p}_i \gamma_i) \right) \\
		\text{s.t.} \quad 
		& \sum_{i=1}^{D} \overline{p}_i \le P, \\
		& \overline{p}_i \ge 0,\quad \forall i \in \{1, 2, \dots, D\},
	\end{align}
\end{subequations}
where $P$ is the maximum transmit power, ${a_{n,\varepsilon}}=\frac{{{Q}^{-1}}\left( \varepsilon  \right)}{\sqrt{n}}\log e$, and ${V}_{i}$ is given by (\ref{Vi}). 

When the number of channel uses $n \to \infty$, water-filling (WF)  method obtains the optimal solution to problem (\ref{PT}). However, the WF method no longer holds optimality when $n$ is finite. Similar to the approach used to solve problem (P2.1) in the inner loop of the PACA algorithm, we can use the SCA method to obtain a locally optimal solution to the problem (\ref{PT}), by solving following problem in the $\zeta-$th iteration
\begin{subequations} \label{PTsca}
\begin{align}
	 \underset{\overline{\mathbf{p}}}{\mathop{\min }}\,&\sum\limits_{i=1}^{D}{\left(\overline{V}'_i\left(\overline{p}_i^{(\zeta)}\right){{\overline{p}}_{i}}- {{\log }_{2}}(1+{{\overline{p}}_{i}}{{\gamma }_{i}}) \right)} \\ 
	 {\rm s.t.}&\sum\limits_{i=1}^{D}{{{\overline{p}}_{i}}}\le P, \\
	& \overline{p}_i \ge 0,\quad \forall i \in \{1, 2, \dots, D\}, 
\end{align}
\end{subequations}
where 
\begin{equation}
	\overline{V}'_i\left(\overline{p}_i^{(\zeta)}\right)=\frac{{a_{n,\varepsilon}}{{\gamma }_{i}}}{{{\left(1+\overline{p}_{i}^{(\zeta )}{{\gamma }_{i}}\right)}^{3}}\sqrt{1-\frac{1}{{{\left(1+\overline{p}_{i}^{(\zeta )}{{\gamma }_{i}}\right)}^{2}}}}}
\end{equation}
and $\overline{p}_i^{(\zeta)}$ is the solution in the $(\zeta-1)-$th iteration. Problem (\ref{PTsca}) can be solved optimally by existing
convex optimization solvers such as CVX \cite{cvx}, and the optimal
solution yields the next iteration.

\section{{Asymptotic Capacity Analysis of STCC-MSC}}
\label{appa3}
{When the number of channel uses $n \to \infty$, the finite-blocklength rate penalty vanishes, and the achievable rate approaches the asymptotic capacity. It is shown in \cite{FY} that single-stream STCC and TCC have the same asymptotic capacity. In the following, we prove that the asymptotic capacity of STCC-MSC is independent of the allowable number of streams $D$, and therefore coincides with that of conventional TCC.
}

{When $n \to \infty$, the rate maximization problem in (\ref{P1}) becomes
\begin{subequations} \label{P1C}
	\begin{align}
		\underset{\mathbf{p},\mathbf{S}}{\mathop{\max }}\,&\sum\limits_{d=1}^{D}{ \sum\limits_{i=1}^{{{N}_{S}}}{\mathbf{S}(i,d)\log \left( 1+{{p}_{i}}{{\gamma }_{i}} \right)} } \\ 
		\text{  s}\text{.t}\text{. }&\sum\limits_{i=1}^{{{N}_{S}}}{{{p}_{i}}\le P},\text{    } \\ 
		& {{p}_{i}}\ge 0,\text{ }\forall i\in \mathcal{N}_S, \\ 
		& \mathbf{S}(i,d)=\{0,1\},\text{ }\forall i\in \mathcal{N}_S,\text{ } \forall d\in \mathcal{D}, \\ 
		& \sum\limits_{d=1}^{D}\mathbf{S}(i,d) \le 1, \text{ }\forall i\in \mathcal{N}_S \label{P1C4}. 
	\end{align}
\end{subequations}
}
{
\begin{lemma}\label{l1}
There always exists an optimal solution $(\mathbf{p}^{\star},\mathbf{S}^{\star})$ to problem~(\ref{P1C}) such that
\begin{equation}
	\sum_{d=1}^{D} \textbf{S}^{\star}(i,d)=1,\qquad \forall i\in\mathcal{N}_S,  \label{45}
\end{equation}
i.e., every subchannel is assigned to one stream.
\end{lemma}
}

\begin{IEEEproof}
{	For the optimal solution $(\mathbf{p}^{\star},\mathbf{S}^{\star})$ to problem~(\ref{P1C}), suppose that there exists at least one subchannel $i_0$ satisfying
		\begin{equation}
			\sum_{d=1}^{D}\textbf{S}^{\star}(i_0,d)=0. \label{46}
		\end{equation}
	 Since subchannel $i_0$ is not assigned to any stream, its contribution to the objective function is always zero regardless of the allocated transmit power. Therefore, allocating a positive power $p^{\star}_{i_0}$ to this subchannel can never improve the objective value, and reallocating this power to any assigned subchannel cannot decrease the achievable rate because the logarithm is monotonically increasing with respect to the transmit power. Hence, whenever (\ref{46}) holds, $p^{\star}_{i_0}$ can be set to zero without affecting the objective value.
}

	{	Therefore, another optimal solution $(\mathbf{p}^{\star}, \tilde{\mathbf{S}})$ can always be constructed by allocating subchannel $i_0$ to an arbitrary stream $d_0 \in \{1,\ldots,D\}$, i.e., setting $\tilde{\mathbf{S}}(i_0,d_0)=1$, while keeping all other elements the same as $\mathbf{S}^{\star}$. Under this newly constructed assignment, the objective value remains unchanged because the corresponding power $p_{i_0}^\star=0$. By repeating this construction for all unassigned subchannels, we obtain an optimal solution that satisfies (\ref{45}). The proof is completed.}
\end{IEEEproof}

{
Under the optimality condition established in Lemma~\ref{l1}, the double summation in the objective function of problem~(\ref{P1C}) reduces to a single summation over the subchannels. Consequently, problem~(\ref{P1C}) can be simplified as
\begin{subequations}
	\begin{align}
		\underset{\mathbf{p}}{\max}\;&\sum_{i=1}^{N_S}\log\left(1+p_i\gamma_i\right)\\
		\text{s.t.}\;&\sum_{i=1}^{N_S}p_i\le P,\quad
		p_i\ge0,\ \forall i\in\mathcal N_S,
	\end{align}
\end{subequations}
which is identical to the asymptotic capacity maximization problem for the conventional TCC. Consequently, the optimal value is independent of the allowable number of streams $D$, and STCC-MSC achieves the same asymptotic capacity as the conventional TCC. The optimal power allocation is given by the classical WF solution.
}

\end{document}